\documentclass[12pt]{article}

\usepackage[utf8]{inputenc}
\usepackage[T1]{fontenc}
\usepackage{microtype}
\usepackage[table]{xcolor}
\usepackage{colortbl}

\usepackage{amsmath,amssymb,amsthm}
\usepackage{bm}

\usepackage[margin=1in]{geometry}
\usepackage{setspace}
\usepackage{natbib}

\usepackage[hidelinks]{hyperref}
\hypersetup{
    colorlinks  = true,
    citecolor   = black,
    linkcolor   = black,
    urlcolor    = black
}

\usepackage{booktabs}
\usepackage{threeparttable}
\usepackage{tabularx}
\usepackage{siunitx}

\usepackage{graphicx}
\usepackage{float}

\usepackage{tikz}
\usepackage{pgfplots}
\usepackage{pgfplotstable}
\pgfplotsset{compat=1.18}
\usetikzlibrary{arrows.meta,positioning,calc,patterns,decorations.pathreplacing}

\usepackage[dvipsnames]{xcolor}

\usepackage{enumitem}
\setlist{noitemsep}

\usepackage{ifthen}

\usepackage{titlesec}
\titleformat{\section}
    {\normalfont\large\bfseries}{\thesection}{1em}{}
\titleformat{\subsection}
    {\normalfont\normalsize\bfseries}{\thesubsection}{1em}{}
\titleformat{\subsubsection}
    {\normalfont\normalsize\itshape}{\thesubsubsection}{1em}{}

\usepackage{fancyhdr}
\usepackage{abstract}

\newtheorem{proposition}{Proposition}

\newtheorem{assumption}{Assumption}
\theoremstyle{remark}
\newtheorem{remark}{Remark}

\newcommand{\sym}[1]{\ifmmode^{#1}\else\(^{#1}\)\fi}
\newcommand{\bfbeta}{\boldsymbol{\beta}}
\newcommand{\bfT}{\mathbf{T}}
\newcommand{\bfRF}{\mathbf{RF}}

\title{The Cascade Identity: 2SLS as a Policy Parameter in Capacity-Constrained Settings}

\author{
Niklas Bengtsson\thanks{Department of Economics, Uppsala University. We thank Magne Mogstad, Alex Torgovitsky, Peter Fredriksson, Adam Altmejd, and Edwin Leuven for comments, and participants at the UPFC Workshop on Economic Inequality and Human Capital. Financial support from the Jan Wallander and Tom Hedelius Foundation (P23-0225) is gratefully acknowledged.}
\and 
Per Engström\footnotemark[1]
}

\date{April 10, 2026\\ \medskip \small}

\begin{document}

\maketitle

\begin{abstract}
\vspace{-3em}
\centering
\noindent\begin{minipage}{0.9\textwidth}
\singlespacing
Governments routinely adjust capacity in rationed programs such as university fields, medical training and public housing, where admitting one individual displaces others and triggers chains of reallocation. We show that in such settings, the standard multi-treatment two-stage least squares (2SLS) coefficient identifies exactly the total societal effect of a marginal expansion, including all downstream reallocations. The result is an algebraic identity: under instrument relevance and a single alignment condition, satisfied in centralized admissions systems, the 2SLS coefficient equals the general-equilibrium shadow value of relaxing a capacity constraint, while the single-instrument Wald ratio captures only the direct effect. Their difference recovers the full equilibrium adjustment without additional structure. Monotonicity is not required. The identity extends beyond queue-based allocation to any fixed-supply setting, including competitive markets with price instruments. We apply the framework to two policy questions in Swedish university admissions, where marginal students are allocated across fields through a centralized lottery mechanism. First, revisiting the debate on whether economics and business education erodes prosocial values, we find that the direct effect of expanding business on charitable giving is precisely zero, but expanding the less competitive fields that business students are displaced from has large prosocial effects. Second, analyzing gender-targeted STEM policies, we find that admitting four women to competitive STEM generates one additional male STEM degree through downstream vacancies. Both are general-equilibrium effects invisible to single-instrument methods.
\end{minipage}
\end{abstract}

\textbf{Keywords:} Instrumental variables, multiple treatments, heterogeneous treatment effects, capacity constraints, cascade effects, university admissions, charitable giving, gender quotas

\section{Introduction}

Governments routinely adjust capacity in programs such as university
fields of study, medical training, and public services. In such
settings, admitting one additional individual typically displaces
someone from another program, whose vacated slot is then filled from
its own queue, triggering a chain of reallocations across the system.
The policy-relevant effect of such an expansion is therefore not just
the effect on the marginal entrant, but the total effect of this
reallocation chain.

We show that a standard econometric object---the multi-treatment
two-stage least squares (2SLS) coefficient---identifies exactly this
total effect. The result is an algebraic identity: under instrument
relevance and the condition that the instrument and the policy
operate on the same allocation margin, the 2SLS coefficient equals
the general-equilibrium shadow value of relaxing a capacity
constraint, including all downstream reallocations through the
system. No monotonicity condition is required, no
restriction on the pattern of cross-program substitution is imposed,
and the identity extends beyond queue-based allocation to any
fixed-supply setting, including competitive markets with price
instruments.

The margin-alignment condition holds by construction in centralized
admissions systems, where capacity expansions shift the admission
cutoff and the instrumental variable strategy exploits local
randomness around that same threshold (e.g.,
\citealp{ockert2010, kirkeboen2016, abdulkadiroglu2017, altmejd2021,
bleemer2022}), and more generally in any rationed setting where slots are
filled from queues, including public housing lotteries
\citep{jacob2012}, military draft lotteries \citep{angrist1990},
and oversubscribed public programs
\citep{kline2016, gelber2016effects}.

The mechanism is straightforward. Expanding a program draws
individuals from other programs, which in turn refill their vacancies
from their own queues, propagating the shock through the system. The
first-stage matrix already encodes these reallocations---how shocks
to one treatment shift enrollment across all others---and 2SLS
aggregates them into their total outcome effect by inverting this
matrix. The off-diagonal elements, often viewed as a complication for \emph{individual}
causal interpretation, are precisely the vacancy-creation rates
through which the policy operates.

The result provides a different way of interpreting multi-treatment IV. With multiple treatments and heterogeneous effects, 2SLS coefficients are typically viewed as weighted averages of individual treatment effects, where the weights may be negative and depend on complex substitution patterns, making the estimand difficult to interpret \citep{mogstad2024handbook}. The standard response has been to impose additional structure on choice behavior, such as latent index models \citep{heckman2008multiple}, observed next-best alternatives \citep{kirkeboen2016}, additive separability \citep{heckman2018}, or marginal treatment response approaches \citep{mogstad2024policy}. More recently, \citet{bhuller2024} show that positive weighting requires ruling out cross-treatment spillovers.

We instead study environments in which such spillovers are intrinsic: capacity-constrained systems in which individuals are allocated across competing programs. In such settings,
the relevant object for welfare analysis is the total societal effect
of a marginal capacity expansion, and we show that this is what the
2SLS coefficient identifies. It is not an average of individual
treatment effects but a general-equilibrium policy parameter: the
marginal societal value of expanding a treatment in a
capacity-constrained system. Under homogeneous effects, the downstream cascade
is inert and this parameter reduces to the standard individual
treatment effect; under heterogeneity, they diverge.

\citet{haavelmo1943} showed that recovering the effect of a policy intervention requires inverting the full simultaneous system, not reading off the reduced form. The modern
literature has developed this theme in capacity-constrained settings
specifically. 
\citet{kline2016} show that in Head Start, the policy-relevant
effect of expanding a program includes the downstream reallocations
it triggers through competing alternatives: an equilibrium object
that the standard Wald ratio misses and that recovering requires a
structural model of the allocation mechanism.
\citet{gandil2025trickle} and \citet{arkhangelsky2025} develop
methods to recover such equilibrium effects; the former by
simulating counterfactual equilibria through a re-engineered version
of the Danish admissions mechanism, the latter by constructing an
adjusted outcome that appends each agent's equilibrium externality,
estimated from LATEs at admission cutoffs. Both approaches require
the researcher to build new objects beyond standard regression
output.

Our result complements these contributions by showing that in any
fixed-supply system, the standard multi-treatment 2SLS coefficient
\emph{already is} the equilibrium-adjusted estimand; no
transformation of the outcome, no simulation of the mechanism, and no
assumptions on individual choice behavior are needed. The result
connects to \citet{carneiro2010}, who show that marginal policy
effects can be identified under weaker conditions than average
effects, without extrapolating beyond the support of the instrument.
The cascade identity is the multi-treatment analogue of this result in capacity-constrained systems, where the relevant margin is defined by the admission cutoff and the policy effect propagates through competing programs rather than operating on a single treatment. A
researcher needs only the 2SLS coefficient $\beta_k$ and the
marginal cost of capacity to conduct cost-benefit analysis. The Wald
ratio $W_k$ provides the partial-equilibrium benchmark; the cascade
correction $\beta_k - W_k$ is the general-equilibrium adjustment.
Both ingredients are standard regression outputs.

We illustrate the framework with two applications using Swedish
university admissions, where marginal applicants are allocated across fields
through a centralized lottery mechanism. The first revisits the
long-standing question of whether economics and business education
erodes prosocial values. We find that the direct effect of expanding
business on charitable giving is precisely zero: the marginal
business student gives no more or less to charity than she would have
in her next-best alternative. But the total policy effect is positive
and significant, driven entirely by the cascade: opening a business
seat draws a student away from a field that would have increased her giving, and refilling that vacancy draws
someone in from outside higher education. The apparent prosociality
gap between business students and others is not caused by business
education; it reflects the prosocial effects of the fields that
business students would otherwise have attended.

The second application analyzes gender-targeted STEM policies.
Quotas and targeted expansions are among the most common
instruments for increasing female representation in competitive
fields, yet their system-level consequences are difficult to
evaluate: admitting a woman to a competitive program displaces
someone from a fallback program, and the vacancy left behind is
filled from a mixed-gender queue. The cascade framework is informative about exactly this type of intervention, since it traces the full
chain of reallocations triggered when the composition of the
marginal entrant changes. We find that admitting one additional
woman to competitive STEM generates 0.25 STEM degrees, of which
roughly one-fifth accrues to men: the woman vacates a mid-tier STEM
slot that is refilled from a predominantly male queue. A
replacement policy---admitting a woman in place of the marginal
man---redistributes STEM degrees across genders but barely changes
the total. Both effects are invisible to single-instrument methods
and emerge only through the cascade.

The remainder of the paper proceeds as follows.
Section~\ref{sec:model} presents the policy parameter, econometric
specification, and assumptions, and derives the cascade identity.
Section~\ref{sec:discussion} discusses the assumptions and their
scope. Section~\ref{sec:heterogeneity} examines heterogeneity and
replacement policies. The following sections apply the framework to
field-of-study choice and prosocial behavior
(Section~\ref{sec:application1}) and to gender-targeted STEM policies
(Section~\ref{sec:stem_application}). Section~\ref{sec:conclusion}
concludes.

\section{The Cascade Identity}
\label{sec:model}

\subsection{The policy parameter}
\label{sec:policy_parameter}

There are $K$ programs, indexed by $j \in \{1,\dots,K\}$, and an
outside option $j = 0$. Each program has fixed and binding capacity. Individual~$i$
assigned to alternative~$j$ generates a specific societal value.
A planner considering whether to expand program~$k$ by one slot needs
the total expected societal effect of this expansion, including all downstream
reallocations. Denote this object~$T_k$.

The marginal entrant to program~$k$ generates a direct benefit~$W_k$ in expectation.
But she may have come from another program~$j$, creating a vacancy
there. When program~$j$ refills that vacancy from its ranking list,
the new entrant to~$j$ may herself have come from program~$m$,
creating a further vacancy. The cascade continues until a slot is
filled from the outside option.

Let $r_{kj}$ denote the \emph{vacancy creation rate}: the number of
vacancies created in program~$j$ per new admission to program~$k$.
Then $T_k$ satisfies the recursion
\begin{equation}
T_k = W_k + \sum_{j \neq k} r_{kj}\, T_j, \qquad k = 1,\dots,K.
\label{eq:cascade_recursion}
\end{equation}
Each reallocation in program~$j$ carries societal value~$T_j$---the
same object we are defining---so the equation is a fixed-point
condition: the value of expanding program~$k$ equals the direct
effect on the marginal entrant plus the value of all induced
reallocations, each of which is itself an equilibrium object defined
by the same primitives.

Equation~\eqref{eq:cascade_recursion} is not an identifying
assumption. It is a consequence of three primitives: (i)~the
definition of~$T_k$ as the total societal effect of a one-slot
expansion, (ii)~fixed total supply of every program other than~$k$,
and (iii)~the allocation mechanism fills vacancies from ranked queues.

It is useful to express (\ref{eq:cascade_recursion}) in matrix form. Collect the vacancy creation rates into a $K \times K$ matrix~$M$
with entries $m_{kj} = r_{kj}$ and zeros on the diagonal. Let
$\mathbf{W} = (W_1, \dots, W_K)^T$ denote the vector of direct
effects. The cascade recursion becomes
\begin{equation}
(I - M)\,\bfT = \mathbf{W},
\label{eq:cascade_fixedpoint}
\end{equation}
with solution $\bfT = (I - M)^{-1}\,\mathbf{W}$ whenever the
spectral radius of~$M$ is less than one.

The Neumann series expansion
\begin{equation}
(I - M)^{-1} = \sum_{n=0}^{\infty} M^n
= I + M + M^2 + M^3 + \cdots
\label{eq:neumann_main}
\end{equation}
has a direct economic interpretation. Each power of~$M$ corresponds
to one round of the cascade:
\begin{itemize}[nosep]
\item $M^0 = I$: the direct effect on the marginal entrant (the
      vector~$\mathbf{W}$).
\item $M^1$: first-round vacancy refills in all other programs.
\item $M^2$: feedback from those refills back through the system.
\item $M^n$: the $n$th round of adjustment.
\end{itemize}
This is Walrasian t\^atonnement: a supply shock propagates through
successive rounds of reallocation until all vacancies are filled.
Convergence of the series, given spectral radius of~$M$ less than
one, is stability of the equilibrium.
Equation~\eqref{eq:cascade_fixedpoint} computes the infinite
t\^atonnement in closed form. Appendix~\ref{app:neumann} works through an instructive two-program case in detail.

The policy parameter~$T_k$ is therefore the general-equilibrium
shadow value of relaxing the capacity constraint in program~$k$.
The direct effect~$W_k$ is the partial-equilibrium shadow value;
the cascade correction $T_k - W_k$ is the general-equilibrium
adjustment. The question is whether standard econometric tools can
recover~$T_k$.

\begin{remark}
    (Reallocation) When the planner cannot expand total capacity but can reallocate a
slot from program~$m$ to program~$k$, the net societal effect is
simply $T_k - T_m$. The contraction of~$m$ by one slot triggers a
vacancy cascade governed by the same matrix~$M$, but in reverse:
individuals spill out rather than in. Since $(I-M)\mathbf{T} =
\mathbf{W}$ continues to hold, the welfare gain of redistribution is
$T_k - T_m$. No new objects need to be estimated; the same
general-equilibrium shadow values~$\mathbf{T}$ that govern expansion
also govern redistribution.
\end{remark}

\subsection{Econometric specification}
\label{sec:econometric}

Consider the linear model
\begin{equation}
Y_i = \beta_0 + \sum_{j=1}^{K} \beta_j A_{ij} + \varepsilon_i,
\label{eq:secondstage}
\end{equation}
estimated by two-stage least squares using instruments
$(Z_{i1},\dots,Z_{iK})$. $A_{ij}$ indicates admission of individual $i$ to programme $j$. Predetermined covariates and fixed effects are suppressed;
all results extend immediately.

Define the \emph{first-stage matrix} $\Pi$ with entries
\begin{equation}
\pi_{jk} \;=\; \frac{\partial\, E[A_{ij}]}{\partial Z_{ik}},
\end{equation}
so that $\pi_{jk}$ measures how instrument~$k$ shifts enrollment in
program~$j$.

The reduced form satisfies
\begin{equation}
\mathrm{RF}_k \;=\; \frac{\partial\, E[Y_i]}{\partial Z_{ik}}
\;=\; \sum_{j=1}^{K}\beta_j\,\pi_{jk},
\end{equation}
or in matrix form
\begin{equation}
\mathbf{RF} \;=\; \Pi^T\boldsymbol{\beta}.
\label{eq:rf_matrix}
\end{equation}

The 2SLS estimand is $\boldsymbol{\beta} = (\Pi^T)^{-1}\mathbf{RF}$:
the reduced form, passed through the inverse of the full first-stage
matrix.

\subsection{Assumptions}
\label{sec:assumptions}

We require two conditions linking the econometric objects to the
planner's problem.

\begin{assumption}[Instrument relevance]
\label{ass:relevance}
$\pi_{kk}\neq 0$ for all $k=1,\dots,K$.
\end{assumption}

\begin{assumption}[Ranking-list refill]
\label{ass:margin}
Any change in program~$j$'s enrollment is mediated by admitting or not
admitting the marginal applicant on program~$j$'s ranking list. The
instrument and the policy of expanding capacity by one slot operate on
this same margin.
\end{assumption}

Assumption~\ref{ass:margin} is an institutional condition satisfied by
construction in centralized admission systems (Sweden, Norway, Denmark,
Chile), where both the instrument and capacity expansions shift the
cutoff along the same ranked queue. It applies equally to school choice
mechanisms with waitlists and other settings where programs fill to
capacity from ranked lists. In less centralized settings, the condition
need not hold by construction but remains a plausible interpretation of
the instrument whenever the researcher treats the IV estimate as
informative about capacity policy.

Assumption~\ref{ass:margin} is thus not a new requirement imposed by the
cascade framework. It is the same condition that makes any LATE estimate
in these settings interpretable as a policy-relevant parameter. When
\citet{kirkeboen2016} argue that their estimates are ``informative about
policy that (marginally) changes the supply of slots in different
fields,'' they are asserting precisely that the instrument and the policy
operate on the same margin, i.e. that the person who would be admitted under
a marginal expansion is the same type of person whose admission is
shifted by the instrument. Without this alignment, the LATE identifies
the effect on lottery compliers but says nothing about the effect of
expanding capacity, because the marginal expansion might admit a
different person through a different channel. Any researcher who
interprets a lottery- or cutoff-based IV estimate as informative about
capacity policy is implicitly invoking Assumption~\ref{ass:margin}. We
merely formalize what the literature already assumes.

A prominent view in applied microeconomics holds that the reduced
form is more
policy-relevant than the IV coefficient, because it captures the effect
of the instrument itself and requires fewer assumptions
\citep{angristpischke2009, imbens2014}. Assumption~\ref{ass:margin} accepts the
premise of this argument, i.e. that the instrument is a relevant
policy, but overturns the conclusion. When the policy is a capacity
expansion that propagates through the system, the reduced form captures
only the effect of one lottery draw on the lottery winner: a
partial-equilibrium object that ignores all downstream reallocations.
It is the 2SLS coefficient---obtained by inverting the full first-stage
matrix---that traces the cascade and recovers the general-equilibrium
policy effect. The distinction between the reduced form as a prediction tool and the structural equation as a policy tool dates to \citet{haavelmo1943}: when the government intervenes by changing a constraint, the structural parameters govern the response---not the conditional expectations from single equations. 

The key implication of Assumption~\ref{ass:margin} is that the
instrument-based first stage identifies the vacancy creation rates
and direct effects from the planner's problem:
\begin{equation}
r_{kj} = \frac{-\pi_{jk}}{\pi_{kk}}, \qquad
W_k = \frac{\mathrm{RF}_k}{\pi_{kk}}.
\label{eq:bridge}
\end{equation}
The logic is simple. Per unit increase in~$Z_k$, program~$k$ gains
$\pi_{kk}$ students and program~$j$ loses $|\pi_{jk}|$ students.
Rescaling to one new admission to~$k$, the number of vacancies
created in~$j$ is $-\pi_{jk}/\pi_{kk}$. Similarly, the reduced
form $\mathrm{RF}_k$ is the total outcome change per unit increase
in~$Z_k$, so the outcome change per new admission---the direct
effect---is the Wald ratio $\mathrm{RF}_k / \pi_{kk}$. These are the vacancy
creation rates and direct effects that enter the cascade recursion
of Section~\ref{sec:policy_parameter}; Assumption~\ref{ass:margin}
ensures they are identified by the instrument variation.
The same marginal responses that govern the effect of the instruments
also govern the effect of a marginal capacity expansion. This bridges
the policy parameter defined in Section~\ref{sec:policy_parameter}
and the econometric objects defined in Section~\ref{sec:econometric}.

Assumption~\ref{ass:margin} does impose a stability requirement: the
ranking list that governs who fills a vacancy is not itself affected by
the expansion. For a marginal change, this is innocuous -- no applicant or
program responds to a single-slot shift. For larger reforms, however, it
rules out endogenous responses on both sides of the market: applicants
must not adjust their application or ranking decisions, and programs must
not alter admission criteria or capacity in other dimensions.
\citet{gandil2025trickle} illustrates the importance of this condition:
when applicants are allowed to adjust their application lists, simulated
ripple effects increase substantially. The cascade identity therefore
characterizes the marginal expansion at current capacity levels; applying
it to non-marginal reforms requires the first-stage matrix to remain
stable, a substantively stronger assumption.

The framework imposes little additional structure beyond these two
assumptions. Treatment indicators need not be mutually exclusive: a
student can be ``ever admitted'' to multiple programs, and the cascade
operates on net enrollment flows. No monotonicity condition is required.
Cross-effects are allowed: $\pi_{jk} \neq 0$ for $j \neq k$, as they
generically are in capacity-constrained systems. Flows need not balance:
$\sum_j \pi_{jk} \neq 0$, so instruments may draw individuals both from
other programs and from the outside option. Indeed, at least some instruments \emph{must} draw entrants from the outside
option: if all column sums $\sum_j \pi_{jk}$ were zero, the
first-stage matrix would be singular and the 2SLS coefficient would
not exist.

\subsection{Main result}
\label{sec:mainresult}

\begin{proposition}[Cascade Identity]
\label{prop:main}
Under Assumptions~\ref{ass:relevance} and~\ref{ass:margin}, the 2SLS
coefficients equal the societal policy effects:
\begin{equation}
\boxed{\;\mathbf{T}=\boldsymbol{\beta}\;}
\end{equation}
That is, $\beta_k$ equals the total societal effect of expanding
program~$k$ by one slot, including all cascading reallocations through
the system.
\end{proposition}

\begin{proof}
By Assumption~\ref{ass:margin}, the first stage identifies the
vacancy creation rates and direct effects:
\begin{equation}
r_{kj} = \frac{-\pi_{jk}}{\pi_{kk}}, 
\qquad
W_k = \frac{\mathrm{RF}_k}{\pi_{kk}}.
\end{equation}

Substituting into the cascade recursion
\eqref{eq:cascade_recursion} gives, for each $k$,
\begin{equation}
\label{eq:cascade}
T_k 
= \frac{\mathrm{RF}_k}{\pi_{kk}}
- \sum_{j \neq k} \frac{\pi_{jk}}{\pi_{kk}}\, T_j.
\end{equation}

Multiplying both sides by $\pi_{kk}$ and rearranging,
\begin{equation}
\mathrm{RF}_k
= \sum_{j=1}^K \pi_{jk}\, T_j.
\end{equation}

Stacking these equations across $k$ yields the matrix system
\begin{equation}
\mathbf{RF} = \Pi^T \mathbf{T}.
\end{equation}

But by the reduced-form identity \eqref{eq:rf_matrix},
\begin{equation}
\mathbf{RF} = \Pi^T \boldsymbol{\beta}.
\end{equation}

Since $\Pi$ is diagonal nonzero (Assumption~\ref{ass:relevance}),
it follows immediately that
\begin{equation}
\mathbf{T} = \boldsymbol{\beta}.
\end{equation}
\end{proof}

The cascade recursion~\eqref{eq:cascade_recursion} is derived from the
primitives fixed supply and queue-based reallocation. It is not assumed to
match the econometric objects. To verify that the identity is not an
artifact of the algebraic setup, Appendix~\ref{app:simulation}
computes~$T_k$ independently by brute force: expanding a program by one
slot in a fully simulated multi-round admissions system with retakes,
programme switching, new cohorts, noncompliance, and defiers. It then sums up the resulting outcome changes across all individuals. The 2SLS
coefficient $\hat{\beta}_k$, estimated on the same simulated data, agrees
with this experimentally computed~$T_k$ to within sampling error (paired
$t$-statistics of 1.35 and 0.60 across 1{,}000 replications).

\begin{remark}[Generality beyond queues]
\label{rem:generality}
The proof uses only two ingredients: the reduced-form identity
$\mathbf{RF} = \Pi^T\boldsymbol{\beta}$ and the cascade
equation~\eqref{eq:cascade}. The first is definitional. The second
follows from any allocation mechanism in which (i)~total supply of each
good is fixed and (ii)~the instrument is the variable through which a
marginal supply expansion is transmitted to individual allocations.
Ranked queues are one such mechanism; competitive markets with price
instruments are another. Appendix~\ref{app:general} proves the cascade
identity for an arbitrary allocation mechanism, showing that the
fixed-supply condition, and not the institutional form, is the essential
assumption.
\end{remark}

\begin{remark}[Potential outcomes]
The cascade identity can also be derived in the potential outcomes framework. The notation is more cumbersome without matrices, but Appendix~\ref{sec:threeprog} works through a two-program (plus outside option) example with ordered treatments, connecting each term in the 2SLS coefficient to a specific complier transition.
\end{remark}

\subsection{Homogeneous treatment effects}
\label{sec:homogeneous}

Under homogeneous treatment effects, every individual gains
$\Delta_j = Y(j) - Y(0)$ from enrolling in program~$j$, regardless of
type. Since the average treatment effect and the marginal policy
effect coincide when effects are constant, the homogeneous case
provides a useful benchmark for interpreting the cascade identity.
 
Consider expanding program~$k$ by one slot. The new entrant may have
come from another program~$j$. Her net gain is
$\Delta_k - \Delta_j$: she gains $\Delta_k$ from entering~$k$ but
loses $\Delta_j$ from leaving~$j$. Her vacated seat in~$j$ is filled
by the next person on~$j$'s ranking list, who may in turn have come
from program~$m$, gaining $\Delta_j - \Delta_m$. This person's
departure from~$m$ is again filled from~$m$'s ranking list, and so
on. The cascade terminates when a seat is filled by someone from the
outside option, who gains $\Delta_\ell$ for whatever program~$\ell$
she enters.
 
Under homogeneity, the intermediate terms telescope. Summing along the
chain:
\[
(\Delta_k - \Delta_j) + (\Delta_j - \Delta_m) + \cdots + \Delta_\ell
= \Delta_k.
\]
Every intermediate program cancels: it appears once as a gain (for the
person entering) and once as a loss (for the person leaving). The only
term that survives is $\Delta_k$, the effect at the origin of the
chain. The cascade is a zero-sum reallocation at every intermediate
step, and the sole net effect is that one additional person, wherever
she sits at the end of the chain, has been drawn into the system
through program~$k$'s expansion.\footnote{When treatment indicators
are not mutually exclusive, the terminal entrant in the cascade need
not come from the outside option; she may already be enrolled in
other programs. However, the second-stage equation is linear and
additive in the treatment indicators, so under homogeneous effects the
incremental gain from adding program~$k$ is $\Delta_k$ regardless of
which other programs the individual is already enrolled in. The
telescoping argument goes through unchanged.} Hence
$\beta_k = \Delta_k = Y(k) - Y(0)$, the standard textbook
interpretation.
 
This makes precise what heterogeneous treatment effects add. When
effects differ across individuals, the person who enters program~$j$
at one link of the cascade may gain more or less than the person who
left~$j$ at the previous link. The intermediate terms no longer
cancel. The 2SLS coefficient $\beta_k$ captures the full chain of
non-cancelling gains and losses. The cascade identity
$\bfT = \bfbeta$ says that 2SLS performs this accounting
automatically: it aggregates the direct effect and all downstream
spillovers into a single policy-relevant number.
 
\begin{remark}[What heterogeneity does]
Under homogeneous effects, the cascade is inert---every reallocation
is zero-sum---and $\beta_k$ reduces to the individual-level causal
effect. Under heterogeneous effects, the cascade is active: each link
contributes a non-zero net effect because the person entering a
program differs from the person leaving. The magnitude of the cascade
correction depends on (i)~the extent of cross-program reallocation
(the off-diagonal elements of $\Pi$) and (ii)~the degree of treatment
effect heterogeneity across complier types at each margin. When either
is small, $\beta_k \approx \Delta_k$; when both are large, the
cascade can substantially alter the policy-relevant effect.
\end{remark}

\section{Discussion}
\label{sec:discussion}

\subsection{The cascade as a general-equilibrium object}
\label{sec:PEtoGE}

The distinction between the Wald ratio $W_k$ and the 2SLS coefficient
$\beta_k$ is most naturally understood as a distinction between partial
and general equilibrium effects. The Wald ratio is the effect on the
lottery winner; the cascade correction $\beta_k - W_k$ is the effect on
everyone else. In this sense, $\beta_k$ is the general-equilibrium
shadow value of relaxing the capacity constraint in programme~$k$,
while $W_k$ is the partial-equilibrium shadow value.

A demand-system analogy makes this transparent. Suppose $K$ goods are supplied in fixed
quantities and prices $p = (p_1,\dots,p_K)$ clear the markets. The
first-stage matrix $\Pi$ corresponds to the matrix of demand
derivatives, with entries $\pi_{jk} = \partial E[Q_j]/\partial p_k$.
See Appendix~\ref{app:general} for details.

A marginal increase in the supply of good~$k$ requires a reduction in
its price $p_k$ to clear the market. The welfare gain from this
price change, holding all other prices fixed, is exactly the Wald
ratio $W_k$. This is a partial-equilibrium object: it ignores that
the change in $p_k$ shifts demand in all other markets through the
cross-price effects $\pi_{jk}$.

These cross-effects require further price adjustments in every other
market. Each adjustment has its own welfare effect and feeds back
into the system through additional cross-price responses. The
economy converges to a new equilibrium through a sequence of such
adjustments.

This is Walrasian t\^atonnement in its standard market economy form. The cascade identity $\bfT = \bfbeta$ then says: in a
fixed-supply economy, the 2SLS coefficient is the
general-equilibrium shadow value of relaxing the supply constraint.
The Wald ratio is the partial-equilibrium shadow value; the cascade
correction is the GE adjustment.

\subsection{Relationship to individual-level LATE estimates}
\label{sec:late_comparison}

The cascade identity targets a general-equilibrium policy effect,
whereas the existing literature on multi-treatment IV targets
partial-equilibrium individual treatment effects. The distinction
maps directly into the PE/GE interpretation of Section~\ref{sec:PEtoGE}.

\citet{kirkeboen2016} and \citet{heinesen2024instrumental} identify
local average treatment effects by conditioning on each individual's
next-best alternative and imposing an irrelevance condition: if
instrument~$k$ does not shift an individual into treatment~$k$, it
does not shift her into any other treatment either. Together, these
assumptions imply that within each next-best stratum, the off-diagonal
elements $\pi_{jk}$ are zero.

With no cross-program spillovers, the policy expansion does
not propagate the shock beyond programme~$k$, and the distinction
between partial and general equilibrium disappears. In this case,
the Wald ratio and the 2SLS coefficient coincide,
\[
\beta_k = W_k,
\]
and both recover the same partial-equilibrium effect.

The cascade identity instead allows for non-zero cross-terms and explicitly
incorporates the resulting spillovers. The off-diagonal elements
that \citet{kirkeboen2016} and \citet{heinesen2024instrumental}
eliminate are exactly the cross-program flows that transmit the
shock through the system. Their approach shuts down general-equilibrium
adjustment; ours traces it out by inverting the full matrix
$\Pi$.

The two approaches therefore answer different questions. For
individual-level pairwise LATEs, the irrelevance and next-best
assumptions are required and our result offers no shortcut. For the
system-level effect of a capacity expansion, these assumptions are
not needed: the first-stage matrix already contains the information
required to recover the general-equilibrium response.

\subsection{The marginal value of public funds}
\label{sec:mvpf}

The cascade interpretation gives the 2SLS coefficient a direct welfare
meaning. The marginal value of public funds for expanding program $k$ is
\begin{equation}
MVPF_k = \frac{T_k}{c_k} = \frac{\beta_k}{c_k},
\end{equation}
where $c_k$ is the marginal cost of a slot. The numerator $\beta_k$ is 
the 2SLS coefficient from the main specification. This object is the multi-treatment analog of the marginal policy 
effect in \citet{carneiro2010} and connects directly to the MVPF 
framework of \citet{hendren2020}: the welfare consequence of an 
infinitesimal expansion of program $k$, expressed per unit cost, 
identified from local variation at the admission margin without 
extrapolation across the full distribution of treatment effects. 
In the spirit of \citet{chetty2009}, the 2SLS coefficient serves 
as a sufficient statistic for welfare analysis in 
capacity-constrained systems; the equilibrium adjustment is 
already embedded in the estimator, requiring no additional 
structural modeling.

This contrasts with approaches such as \citet{kline2016} and \citet{gandil2025trickle}, where the
Wald ratio identifies only the direct effect on the lottery winner, and
the analyst must model or simulate the downstream re-allocations, and
their fiscal consequences, to recover the policy-relevant effect.
In the cascade framework, these equilibrium adjustments are already
embedded in the 2SLS coefficient. The researcher needs only the marginal
cost of capacity to conduct cost-benefit analysis; no additional
structural modeling of reassignment flows or costs in other programs is
required.

The result highlights a key advantage of the cascade interpretation:
while individual-level LATE parameters are often difficult to map into
policy-relevant welfare objects, the multi-treatment 2SLS coefficient
directly identifies the marginal policy effect needed for cost-benefit
analysis in capacity-constrained systems.

\section{Subsample effects in a capacity-constrained system}
\label{sec:heterogeneity}

A natural question in any treatment evaluation is whether effects
differ across subgroups, for example, by gender, age, or family
background. In the standard LATE framework with a single uncapped
treatment, the approach is straightforward: estimate the model
separately on each subgroup. The resulting coefficient identifies the
LATE for compliers within that subgroup. The cascade framework can
address the analogous question: what is the societal effect of a
group-targeted expansion, such as reserving an additional slot for
women? But the system-level nature of the estimand requires a
decomposition that accounts for how the cascade propagates through
mixed-group queues.

\subsection{Why subsample estimation answers the wrong question}

Consider estimating $\beta_k$ on women only, i.e. dropping all men from
the sample and running 2SLS with the female subsample. The resulting
coefficient $\beta_k^f$ inverts the women-only first-stage matrix
$\Pi^f$ and uses the women-only reduced form. By the cascade identity
applied to this subsample, $\beta_k^f$ equals the societal effect of
expanding programme $k$ by one \emph{women-only} slot, where the
entire downstream cascade also operates exclusively through women.

This is a fictional policy. The system does not have a gender. In the real admission system, when a woman
vacates a slot in programme $j$, the next person on $j$'s ranking list
may be a man. The cascade is gender-blind from the second step onwards.
The women-only estimate imposes a single-gender system that does not
exist, producing a parameter that does not correspond to any
implementable policy.

Two well-defined heterogeneity questions can be asked instead. We
describe each in turn.

\subsection{Decomposing the policy effect by subgroup}

The first question is: \emph{when programme $k$ expands by one slot,
how much of the total societal effect accrues to women versus men?}

This is answered by running full-sample 2SLS with group-specific
outcomes as dependent variables. Define $Y_i^f = f_i \cdot Y_i$ and
$Y_i^m = (1-f_i) \cdot Y_i$, where $f_i$ is a female indicator.
Estimate
\begin{equation}
Y_i^f = \alpha_0^f + \sum_{j=1}^K \beta_j^{f*} A_{ij} + \varepsilon_i^f
\end{equation}
by 2SLS on the full sample, using the same instruments and endogenous
variables as the main specification. Define $\beta_j^{m*}$ analogously
from the regression with $Y_i^m$ as the dependent variable.

Since $Y_i^f + Y_i^m = Y_i$, linearity of 2SLS implies
\begin{equation}
\beta_k^{f*} + \beta_k^{m*} = \beta_k
\end{equation}
for each $k$. The decomposition is exact and additive. Both
$\beta_k^{f*}$ and $\beta_k^{m*}$ inherit the cascade interpretation
from the full-sample specification: the first-stage matrix $\Pi$ is
the actual mixed-gender system, and the cascade operates through
mixed-gender queues at every step. We are simply partitioning the
outcome changes at each step into those experienced by women and
those experienced by men.

The coefficient $\beta_k^{f*}$ answers: ``of the total societal effect
of expanding programme $k$ by one slot, how much is due to changes in
women's outcomes?'' This includes women who enter programme $k$
directly, women who enter other programmes through the cascade, and
women who are displaced. It is a well-defined policy parameter that
requires no fictional single-gender system.

This approach generalises immediately to any partition of the
population: age groups, parental income quartiles, prior education
levels. For any exhaustive set of group indicators $g \in \{1,\dots,G\}$
with $\sum_g \mathbf{1}[g_i = g] = 1$, define $Y_i^{(g)} =
\mathbf{1}[g_i = g] \cdot Y_i$ and estimate full-sample 2SLS for
each group outcome. The coefficients sum to the total:
$\sum_g \beta_k^{(g)*} = \beta_k$.

\subsection{The effect of conditioning the marginal entrant}

The second question is: \emph{what is the total societal effect of
admitting one more woman (specifically) to programme $k$, rather than
one more person of unspecified gender?}

This question asks whether it matters for society who the marginal
entrant is. It is the closest analogue to the standard heterogeneous
LATE and requires a different construction.

When programme $k$ admits one more woman:
\begin{itemize}[nosep]
\item The direct effect is on a woman. Her outcome changes by the
      female-specific treatment effect at the margin, captured by the
      women-only reduced form and first stage: $\mathrm{RF}_k^f /
      \pi_{kk}^f$.
\item She may vacate a slot in programme $j$. The probability of this
      is determined by the women-only first stage: $-\pi_{jk}^f /
      \pi_{kk}^f$.
\item The vacated slot in $j$ is refilled from $j$'s mixed-gender
      ranking list. From this point on, the cascade is gender-blind.
      The societal effect of one additional slot in $j$ is the full-sample $\beta_j$.
\end{itemize}

The total effect of admitting one more woman to programme $k$ is
therefore
\begin{equation}
T_k^{|f} = \frac{\mathrm{RF}_k^f}{\pi_{kk}^f}
+ \sum_{j \neq k} \frac{-\pi_{jk}^f}{\pi_{kk}^f} \, \beta_j,
\label{eq:conditional_het}
\end{equation}
where $\mathrm{RF}_k^f$ and $\pi_{jk}^f$ are the reduced form and
first stage estimated on women only, and $\beta_j$ is the full-sample
2SLS coefficient.

The first term is female-specific: the direct effect on the marginal
woman that fills the new slot. The cascade correction uses female-specific vacancy creation
rates (since it is a woman who potentially departs programme $j$) but
full-population cascade effects (since the refill draws from the
mixed-gender queue). All ingredients are standard output
from two sets of 2SLS regressions: the women-only first stage and
reduced form, and the full-sample $\beta_j$. 

Comparing $T_k^{|f}$ with $T_k^{|m}$ (defined analogously for men)
reveals whether the gender of the marginal entrant matters for the
total societal effect. The two can differ for two reasons. First, the
direct effect may differ: the marginal woman admitted to programme $k$
may gain more or less than the marginal man would
($\mathrm{RF}_k^f / \pi_{kk}^f \neq \mathrm{RF}_k^m / \pi_{kk}^m$).
Second, the cascade triggered by a woman may differ from the cascade
triggered by a man, because women and men tend to vacate different
programmes when admitted to $k$ ($\pi_{jk}^f \neq \pi_{jk}^m$). For
example, if women admitted to a STEM programme disproportionately
vacate slots in health sciences while men disproportionately vacate
slots in business, the downstream reallocation differs by gender even though the refill at each
step draws from the same mixed-gender queue. If neither channel is
operative ($T_k^{|f} = T_k^{|m}$), the policy maker expanding
programme $k$ need not be concerned with the gender composition of
the marginal entrants.

\subsection{Replacement policies}
\label{sec:replacement}

The cascade framework extends naturally from capacity expansions to
\emph{replacement policies}: interventions that hold total capacity
fixed but change the composition of entrants. A quota, an affirmative
action rule, or a reallocation of slots across applicant types are all
examples. The key observation is that the marginal policy effect $T_k$
is defined as a derivative, so it is symmetric: expanding program~$k$
by one slot has effect $+T_k$, and contracting by one slot has effect
$-T_k$. A policy that simultaneously admits one type and removes
another is the difference of two conditional policy effects.

Suppose the planner replaces the marginal type-$m$ entrant in
program~$k$ with a type-$f$ entrant, holding total capacity fixed. This
is equivalent to two simultaneous operations: admit one more type-$f$
individual (effect $T_k^{|f}$) and remove one type-$m$ individual
(effect $-T_k^{|m}$). The net effect of the replacement is
\begin{equation}
R_k^{f \leftarrow m}
\;=\; T_k^{|f} - T_k^{|m}.
\label{eq:replacement}
\end{equation}

Each term decomposes using equation~\eqref{eq:conditional_het}:
\begin{align}
T_k^{|f} &= \frac{\mathrm{RF}_k^f}{\pi_{kk}^f}
+ \sum_{j \neq k} \frac{-\pi_{jk}^f}{\pi_{kk}^f} \, \beta_j, \\[4pt]
T_k^{|m} &= \frac{\mathrm{RF}_k^m}{\pi_{kk}^m}
+ \sum_{j \neq k} \frac{-\pi_{jk}^m}{\pi_{kk}^m} \, \beta_j.
\end{align}
The replacement effect $R_k^{f \leftarrow m}$ therefore has two
sources. First, the direct effects differ: the type-$f$ entrant may
gain more or less from program~$k$ than the type-$m$ entrant she
replaces ($\mathrm{RF}_k^f / \pi_{kk}^f \neq \mathrm{RF}_k^m /
\pi_{kk}^m$). Second, the two types trigger different cascades:
type-$f$ and type-$m$ individuals vacate different fallback programs
when admitted to~$k$ ($\pi_{jk}^f \neq \pi_{jk}^m$), so the
downstream reallocations differ even
though both cascades draw from the same mixed-type queues at every
subsequent step.

The replacement effect can be further decomposed by who benefits. For
any partition of the population into groups $g \in \{1,\dots,G\}$,
define group-specific outcomes $Y_i^{(g)} = \mathbf{1}[g_i = g]
\cdot Y_i$ and let $\beta_j^{(g)*}$ denote the full-sample 2SLS
coefficient from regressing $Y_i^{(g)}$ on the treatment indicators.
Then
\begin{equation}
R_k^{f \leftarrow m,(g)}
\;=\; T_k^{|f,(g)} - T_k^{|m,(g)},
\label{eq:replacement_group}
\end{equation}
where
\[
T_k^{|f,(g)} = \frac{\mathrm{RF}_k^{f,(g)}}{\pi_{kk}^f}
+ \sum_{j \neq k} \frac{-\pi_{jk}^f}{\pi_{kk}^f} \, \beta_j^{(g)*}
\]
measures how admitting one more type-$f$ individual to program~$k$
affects group~$g$'s outcomes, and $T_k^{|m,(g)}$ is defined
analogously. By linearity, the group-specific replacement effects sum
to the total: $\sum_g R_k^{f \leftarrow m,(g)} = R_k^{f \leftarrow m}$.

The replacement formula~\eqref{eq:replacement} answers a question that
arises naturally in any capacity-constrained system: what happens when
the planner changes \emph{who} fills a slot rather than \emph{how many}
slots exist? The expansion effect $T_k^{|f}$ and the contraction effect
$-T_k^{|m}$ each include a direct component and a cascade component.
The direct components measure the difference in treatment effects
between the two types at program~$k$'s margin. The cascade components
measure the difference in the downstream reallocations that each type
triggers: which programs they vacate, who fills those vacancies, and
what those individuals would otherwise have done.

The decomposition by beneficiary~\eqref{eq:replacement_group} makes
the distributional consequences of replacement policies transparent. A
gender quota in program~$k$, for instance, replaces the marginal man
with a woman. The group-specific replacement effects reveal how this
swap affects women's and men's outcomes separately, including all effects working through
the cascade. The direct effect of the quota operates on the two
individuals directly involved (the woman who enters and the man who is
displaced). The cascade effect operates on everyone else in the system
who is reallocated as a consequence. If the woman and the man who are
swapped would have attended different fallback programs, the cascade
triggered by the quota differs from the cascade that would have
occurred without it, and the distributional consequences extend well
beyond the two individuals at program~$k$'s margin.

All ingredients are standard regression output. The type-specific first
stages and reduced forms come from subsample IV regressions; the
full-population $\beta_j$ and $\beta_j^{(g)*}$ come from full-sample
2SLS with the overall and group-specific outcomes. No additional
structural modeling is required. The replacement effect is identified
under the same assumptions as the expansion effect: instrument
relevance and ranking-list refill.
\section{Applications: Effects of expanding fields of study}
\label{sec:application}
The cascade estimator is particularly well-suited to field-of-study effects because the
counterfactual for a student marginally admitted to, say, economics, STEM or medicine is
not ``no higher education'' but rather enrollment in whatever program
they would have attended instead. A single-instrument Wald ratio that
ignores this substitution confounds the effect of attending economics
with the effect of not attending a program that may itself increase
prosocial behavior. The cascade identity resolves this by estimating the
full societal effect of expanding each field by one slot, including the
downstream reallocation of students across substitute programs that such
an expansion triggers.

Our application uses Swedish data. In the Swedish university admissions system, applicants submit a single
ranked list of programs and are admitted through a national clearing
process administered by UHR (\emph{Universitets och högskolerådet}). Admission is determined by a merit score
(\emph{meritvärde}) based on upper secondary grades, results from the
Swedish Scholastic Aptitude Test (\emph{Högskoleprovet}), or prior
undergraduate credits. Each program allocates seats across parallel
tracks corresponding to these criteria. When qualified applicants exceed
available seats, the admission cutoff falls at the lowest merit score
among admitted students within each track. Ties at the cutoff are
resolved by lottery, conducted independently within each program and
admission round. The lottery produces a complete ranking of applicants
within each merit-score bracket, with position determined by chance
alone.

The lottery tie-breaking rule defines a natural set of \emph{pivotal
groups}: for each program and admission round, the pivotal group
consists of all applicants whose merit score equals the admission cutoff.
Within this group, admission is determined by the lottery draw alone,
while applicants above the cutoff are admitted with certainty and those
below are rejected with certainty. We construct a luck variable equal to
the applicant's normalized rank within the pivotal group, which is
uniformly distributed between 0 and 1 and independent of all
pre-determined characteristics conditional on group membership.
Crucially, the lottery ranking also determines the marginal applicant in
the natural policy experiment of expanding a program by one seat: it is
the highest-ranked rejected applicant, i.e. the next in the lottery
queue who would be admitted under a marginal capacity expansion. This
alignment between the instrument and the policy margin is what gives the
2SLS estimator its institutional interpretation.

This design is closely related to the regression discontinuity approach
used in related work on Scandinavian higher education
\citep{altmejd2021, kirkeboen2016}, which identifies causal effects by
comparing applicants just above and just below the admission cutoff. Our
setting differs in that applicants within the same merit-score bracket
are \emph{identical} on the running variable by construction, and
admission among them is determined by lottery. Continuity assumptions and
bandwidth choices are therefore not required, and the instrument can be
treated as the outcome of a randomization within each pivotal group.

Not all programs use pure lottery tie-breaking. Until 2012, a small
number of programs applied gender quotas at the admission margin. Other
programs resolve ties using \emph{Högskoleprovet} scores rather than
random draws, generating a deterministic rather than stochastic margin.
Both deviations are flagged in the administrative data. We exclude all
affected programs from the analysis, retaining only programs for which
the within-pivotal-group luck variable is consistent with random
assignment. The final analytical sample comprises 11,604 pivotal groups
across 28 admission rounds between 2008 and 2021, with an average
pivotal group size of approximately 12 applicants.

We group the 11,604 pivotal groups in our sample into seven mutually
exclusive fields using the standard Swedish SUN classification
(\emph{Standard f\"{o}r svensk utbildningsnomenklatur}). The fields are:
business including economics (SUN 34 and 314), social science (SUN 2--3
excluding business), pedagogy (SUN 1), medicine (SUN 721, 724, 727),
health (remaining SUN 7), STEM (SUN 4--5), and a residual other
category. Appendix~\ref{sec:aggregation} discusses the aggregation
weights and the conditions under which field-level estimates inherit the
cascade interpretation from the underlying program-level effects.

The structural equation is
\begin{equation}
Y_i = \beta_0 + \sum_{j=1}^{F} \beta_j A_{ij} + \mathbf{X}_i'\boldsymbol{\gamma}
+ \varepsilon_i,
\label{eq:secondstage_app}
\end{equation}
where $Y_i$ are application-specific outcomes in the years following
the admission decision, $A_{ij}$ indicates enrollment in field $j$, and
$\mathbf{X}_i$ contains field of application, prior giving (application 1), gender, age,
year of admission, and application priority. We estimate the system by
2SLS on the sample of pivotal-group members, using field-specific lottery
instruments $Z_{ij} = L_i \times \mathbf{1}(\text{applied to field } j)$
as excluded instruments, where
\begin{equation}
L_i = 1 - \frac{\text{lottery rank}_i}{(1+n_g)}
\end{equation}
is the applicant's normalized lottery rank within pivotal group $g$ of
size $n_g$. By construction, $L_i$ is symmetric and uniformly distributed
on $(0,1)$ within each pivotal group and therefore orthogonal to all
variables fixed within groups, including field of application, year of
admission, and merit score. Pivotal group fixed effects are consequently
not required for identification. Section~\ref{sec:balanceandrobustness}
verifies balance on predetermined individual characteristics and
demonstrates robustness to alternative control specifications (including
pivotal group fixed effects).

By the cascade identity of Section~\ref{sec:model}, the coefficient
$\hat{\beta}_j$ recovers the full societal effect of expanding field $j$
by one seat, inclusive of all downstream student reallocations across
substitute fields. A central feature of the cascade framework is that
the cascade effect itself requires no additional estimation beyond what
is already standard practice. The difference between the full 2SLS
coefficient $T_k$ and the single-instrument Wald ratio
$W_k = \mathrm{RF}_k / \pi_{kk}$ recovers exactly the downstream
reallocation effect:
\begin{equation}
T_k - W_k = \sum_{j \neq k} \frac{-\pi_{jk}}{\pi_{kk}} T_j.
\end{equation}
Both objects are standard regression output: $T_k$ is the coefficient
from the full 2SLS system and $W_k$ is the coefficient from a
single-instrument IV regression restricted to field $k$'s own lottery.
Their difference, reported in column~(5) of Table~\ref{tab:main},
measures the societal effect of the reallocation triggered by a marginal
expansion of field $k$; the effect on everyone else in the system, net
of the direct effect on the lottery winner.

\subsection{Application 1: Field of Study and Prosociality}
\label{sec:application1}
Universities have long claimed that their purpose extends beyond the
transmission of technical skills to the cultivation of prosocial values,
civic responsibility, and moral character.\footnote{Classical and
progressive educational theory emphasizes education as character
formation and preparation for democratic citizenship; see, e.g.,
\citet[Book VIII]{aristotle_politics} and \citet{dewey1916}. A large
empirical literature studies the causal impact of education in general on
prosociality and civic engagement, typically exploiting compulsory
schooling reforms or other sources of exogenous variation in years of
schooling \citep{dee2004, milligan2004, glaeser2002, helliwellputnam2007,
persson2015, almen2025}. \citet{dee2020} and \citet{willeck2022} provide
recent surveys.} The content of these commitments differs systematically
across fields. Programs in health care, social work, and education frame
their mission around public service and human development, while programs
in economics, business, and engineering emphasize competition, innovation,
and market-based problem solving. Whether these differences in mission
translate into differences in the prosocial behavior of graduates is an
open empirical question with direct implications for education policy.

A long-standing literature documents that economics and business students
exhibit lower levels of cooperative behavior, charitable giving, and
trust relative to students in other fields \citep[see,
e.g.,][]{frank1993, frey2003, bauman2011, carter1991, sundemo2025}.
Two competing explanations have proved difficult to disentangle. The
first is \emph{treatment}: economics education itself shapes values by
reinforcing a model of human behavior premised on narrow
self-interest---the \emph{homo economicus} assumption---or by
legitimizing self-interested action through the language of incentives
and efficiency \citep{marwell1981, frank1993}. The second is
\emph{selection}: students who choose economics are already less
prosocially oriented before enrollment, and the field attracts rather
than creates self-interested individuals \citep{bauman2011,
girardi2024}. Fully resolving the selection-versus-treatment debate
requires exogenous variation in field assignment itself, i.e., precisely what
the Swedish admissions lottery provides.

We measure prosociality using annual data on charitable giving drawn from
Swedish tax registers, which record all donations qualifying for tax
deductions,\footnote{The tax deduction for charitable donations was
available during 2012--2015 and reintroduced from 2019 onward. Giving
data are therefore observed for these years only. Since our admission
sample covers 2008--2021, nearly all applicants have at least one
post-admission observation of giving. However, for applicants admitted
before 2012, we lack a pre-admission measure of giving. To maintain a
common sample across all specifications, we set prior giving equal to
zero for these applicants. Section~\ref{sec:balanceandrobustness} shows
that the main estimates are robust to alternative treatments of this
variable.} for the full population of university applicants over
2008--2021. Giving is observed both before and after the admission
decision, allowing us to control for pre-existing levels and trends and
isolate the causal effect of field attendance. As a revealed-preference
measure that is private, voluntary, and financially costly, it is free from
the social desirability bias that plagues survey-based attitude measures.

Figure~\ref{fig:firststage_heatmap} displays the estimated first-stage
matrix. Each cell reports the effect of a high lottery score in one field
(column) on the probability of eventual admission (i.e.\ ever admitted
during 2008--2021) to the same or another field (row), measured over a
two-to-twelve year window following the lottery draw.\footnote{An applicant is only coded as ever admitted if she was admitted to the field through competition. This means that all admissions that did not have applicants in the final queue, i.e. that were not capacity constrained, are counted towards the outside option instead. Technically, the outside option thus includes a few students that were admitted to university programmes that had surplus capacity and could not fill all its slots.} The diagonal
elements are large, positive, and uniformly significant: a high lottery
score for a given field substantially raises the probability of eventual
admission to that same field, confirming strong compliance. The
off-diagonal elements are negative throughout, consistent with the
capacity-constraint mechanism underlying the cascade identity: a student
who secures admission to one field is displaced from others.

\begin{figure}[!ht]
\centering
\begin{tikzpicture}[x=1.6cm,y=1.3cm]
\def\maxabs{0.50}
\def\opSig{1}
\def\opNS{0.18}
\foreach \j/\lab in {
    0/Business, 1/Social science, 2/Teaching,
    3/Medicine,  4/Other health,    5/STEM, 6/Other}{
  \node[rotate=45, anchor=west] at (\j,7.2) {\lab};
}
\foreach \i/\lab in {
     0/Business, 1/Social science, 2/Teaching,
    3/Medicine,  4/Other health,    5/STEM, 6/Other}{
  \node[anchor=east] at (-0.6,6-\i+0.5) {\lab};
}
\node[font=\bfseries] at (3.5,8.8) {Field applied to};
\node[font=\bfseries, rotate=90] at (-2.4,3.5) {Field admitted to};
\newcommand{\cell}[6]{%
  \pgfmathsetmacro{\intensity}{sqrt(min(abs(#3)/\maxabs,1))*100}%
  \ifdim #3 pt > 0.00001pt
    \def\basecolor{red}%
  \else
    \ifdim #3 pt < -0.00001pt
      \def\basecolor{blue}%
    \else
      \def\basecolor{white}%
    \fi
  \fi
  \ifnum #5=1
    \def\myop{\opSig}%
  \else
    \def\myop{\opNS}%
  \fi
  \filldraw[draw=black!15,
            fill=\basecolor!\intensity!white,
            fill opacity=\myop]
      (#1, 6-#2) rectangle ++(1,1);
  \node[font=\scriptsize, align=center]
      at (#1+0.5, 6-#2+0.5) {#6};
}
\input{heatmap_cells}
\draw[black!30] (0,0) rectangle (7,7);
\end{tikzpicture}
\caption{First-stage matrix. Effect of the instrument (lottery draw) on ever (2008-2021) being admitted to a specific field. Rows indicate fields admitted to ($D_k$) and columns indicate field applied to. Red~=~positive coefficient, blue~=~negative.
Faded color~=~not significant ($|t|<1.96$). $t$-values in parentheses. 149579 observations, t-statistics computed using clustered standard errors on 11,604 pivotal groups.}
\label{fig:firststage_heatmap}
\end{figure}

The pattern of cross-effects reveals the substitution structure of the
Swedish higher education market. Business, social science, STEM, and
medicine exhibit substantial mutual substitution, indicating that
applicants to these competitive fields hold overlapping preferences and
treat them as close alternatives. Teaching and health, by contrast, show
negligible cross-effects with other fields. Their applicant pools are
largely distinct, and the counterfactual for a rejected applicant is more
likely to be no higher education at all rather than enrollment in a
substitute field. This reflects the lower average merit-score cutoffs for
teaching and non-specialist health programs, which draw from a different
part of the applicant distribution than the more competitive fields.
Appendix~\ref{sec:heatmap2} shows a more granular table at the 3-digit
SUN-level, revealing that there is also some substitution within fields.

These cross-effects are the structural inputs to the cascade estimator.
The off-diagonal elements of the first-stage matrix determine the cascade correction that separates the direct effect of
enrolling in a field from the indirect effect of displacing students into
substitute fields. The heatmap thus summarizes the identification
structure of the entire analysis: strong diagonals confirm that each
field's own lottery instrument provides credible variation, while the
pattern of off-diagonal substitution maps the reallocation mechanism
through which a marginal capacity expansion propagates across the system.

Table~\ref{tab:main} reports the main results. Column~(1) reports the
full 2SLS estimate $T_k$, interpreted as the total societal effect of
admitting one additional student to field $k$, including both the direct
effect on the marginal entrant and all downstream reallocation effects
through the capacity-constrained system. Column~(2) reports the
own-instrument Wald ratio $W_k$, which captures only the direct effect
on the lottery winner, ignoring what happens to displaced students.
Column~(5) reports the cascade $T_k - W_k$.

The results reveal substantial heterogeneity across fields. Teaching,
medicine, health, and STEM all generate large and statistically
significant total effects on charitable giving, with point estimates
ranging from 0.045 to 0.077. For teaching, virtually the entire effect is
direct: the cascade is small and insignificant (0.002), indicating that vacancies created by an additional teaching seat do not create prosocial effects. This is consistent with the
first-stage evidence that teaching draws from a distinct applicant pool
with few close substitutes. The counterfactual for a rejected teaching
applicant is more likely no higher education than enrollment in another
prosocial field. For medicine, health, and STEM, by contrast, the
cascade accounts for a meaningful share of the total effect, ranging from
0.01 to 0.02, all statistically significant. Expanding these
competitive fields displaces applicants who would otherwise have enrolled
in other programs with positive prosocial effects, and these programs
recruit, in turn, individuals from outside higher education. This
reallocation adds to the total societal benefit of expanding these
disciplines.

The results for business and social science are particularly instructive
for the debate on economics and prosociality. The own-instrument Wald
ratio for business ($W_k = 0.0042$) is precisely estimated and
statistically indistinguishable from zero. This is a clean
estimate of what attending business does to a student's prosocial
behavior relative to their counterfactual: it finds no effect whatsoever.
Students who are marginally admitted to business programs give no more or
less to charity than they would have had they attended their next-best
alternative. The total policy effect of expanding business with one seat
is, however, positive ($T_k = 0.0196$), and the difference is the cascade
($T_k - W_k = 0.0154$, $p < 0.01$). Opening one additional seat in
business draws a student away from their next-best alternative---social
science, medicine, or another field with a positive prosocial effect---and
filling that vacancy generates a societal gain that the own-instrument
estimate misses entirely.

\begin{table}[htbp]\centering
\caption{The effects of expanding fields of study on charitable giving}
\label{tab:main}

\begin{threeparttable}
\begin{tabular}{lccccc}
\hline
                    & $\beta_k$ & \multicolumn{3}{c}{$W_k$} & $\Delta$ \\
\cline{3-5}
                    &           & $b_k$ & $F$-stat & $N$ & (Cascade) \\

&(1)&(2)&(3)&(4)&(5) \\
\hline
\\
Business            &       .0196         &      .00416         &         427         &       19882         &       .0154\sym{***}\\
                    &     (.0149)         &     (.0132)         &                     &                     &     (.0054)         \\
\\
Social science      &       .0368\sym{**} &        .019         &         651         &       47011         &       .0177\sym{***}\\
                    &     (.0156)         &      (.014)         &                     &                     &    (.00602)         \\
\\
Teaching            &       .0579\sym{***}&       .0561\sym{***}&         449         &       11102         &      .00185         \\
                    &     (.0221)         &     (.0217)         &                     &                     &    (.00351)         \\
\\
Medicine            &       .0767\sym{***}&       .0586\sym{**} &         165         &       21169         &       .0181\sym{**} \\
                    &     (.0266)         &     (.0252)         &                     &                     &    (.00743)         \\
\\
Health              &       .0531\sym{***}&       .0449\sym{**} &         588         &       26867         &      .00815\sym{**} \\
                    &     (.0199)         &     (.0194)         &                     &                     &    (.00377)         \\
\\
STEM                &       .0452\sym{***}&       .0354\sym{**} &         449         &       20358         &      .00979\sym{**} \\
                    &     (.0168)         &     (.0152)         &                     &                     &    (.00498)         \\
\\
Other               &      -.0252         &      -.0349\sym{*}  &         255         &        3190         &       .0097         \\
                    &     (.0201)         &     (.0191)         &                     &                     &    (.00689)         \\
\\
$\bar{\beta}$       &      0.0377         &                     &                     &                     &                     \\
~SE($\bar{\beta}$)  &      0.0099         &                     &                     &                     &                     \\
KP F-stat           &      53.654         &                     &                     &                     &                     \\
Joint $\chi^2$ (betas=0)&      26.790         &                     &                     &                     &                     \\
Observations        &      149579         &                     &                     &                     &                     \\

\\
\hline
\end{tabular}

\begin{tablenotes}[flushleft]
\footnotesize
\item \textit{Notes:} Column (1), $\beta_k$, reports results from a single 2SLS regression with all fields jointly instrumented; columns (2)--(4), $W_k$, report results from separate single-instrument regressions for each field (hence the field-specific first-stage F-statistics and sample sizes). Outcome variable is a binary indicator equal to one if student donated a positive amount to charity in the future years, zero otherwise. Standard errors in parentheses. Standard errors for $\beta_k$ and $b_k$ are clustered at the pivotal group level. Standard errors for the cascade (5) are computed by cluster bootstrap (1000 reps), resampling at the level of pivotal groups. The F-statistic refers to the Kleibergen–Paap first-stage F-statistic, testing the joint significance of the excluded instruments. \sym{*} $p<0.10$, \sym{**} $p<0.05$, \sym{***} $p<0.01$. 
\end{tablenotes}

\end{threeparttable}
\end{table}

Social science presents a similar but empirically sharper pattern. The
total policy effect is positive and significant ($T_k = 0.037$,
$p < 0.05$), yet the own-instrument Wald ratio is half as
large and statistically insignificant ($W_k = 0.019$). The significant
cascade ($T_k - W_k = 0.018$, $p < 0.01$) accounts for the difference.
A researcher reporting only $T_k$ would conclude that expanding social
science raises charitable giving; a researcher reporting only $W_k$
would conclude it has no effect. Both conclusions are correct since they
answer different questions. The cascade decomposition reveals that the
policy effect of expanding social science operates substantially through the
reallocation of displaced students into fields with stronger prosocial
effects, rather than through the direct effect of social science
attendance itself. This is precisely the confound that a
single-instrument Wald estimate cannot disentangle, and that the cascade
framework resolves.

A broader implication for the literature is that estimates of the prosocial ``cost'' of economics or business education may conflate two distinct effects: the impact of studying business and the impact of not studying the alternative. Our results indicate that the latter drives much of the observed differences. In particular, fields such as pedagogy, health, and medicine appear to generate substantial increases in prosocial behavior. From this perspective, economics and business do not so much erode an underlying baseline of prosociality; rather, they fall short of the prosocial gains induced by competing fields.

The field-level estimates in Table~\ref{tab:main} are uniformly positive
(except for ``other'', a small group of diverse programs), suggesting that higher education broadly
increases prosocial behavior. The average $\hat{\beta}_k$ of 0.0377 is
statistically significant, providing causal evidence in favor of the
long-standing view that expanding university education cultivates moral
character beyond the transmission of technical skills (in this sample, the average charitable giving was about 0.1 in 2021; see Table \ref{tab:summary_fields} in Appendix \ref{sec:appendix} for field-specific averages). A natural question
is how this average, which reflects the mean effect of expanding each
field by one slot, relates to what a researcher would typically estimate
when asking whether higher education in general increases prosocial
behavior: regressing the outcome on a single indicator for any
enrollment, instrumented by $L_i$ or a set of field-specific
interactions. This is the standard single-regressor local average
treatment effect (LATE). Appendix~\ref{sec:aggregation} discusses the
relationship between the average of $\beta_1, \dots, \beta_K$ and this
pooled 2SLS estimate, and Table~\ref{tab:aggregation} reports the
results.

\subsection{Application 2: Gender quotas and STEM degree attainment}
\label{sec:stem_application}

A central objective of higher education policy in many countries is to
increase female representation in STEM fields: science, technology, engineering, and
mathematics \citep{blickenstaff2005, ceci2011}. Interventions range from
targeted recruitment and mentoring programs to explicit gender quotas or
preferential admission rules. Two distinct policy instruments are
commonly discussed. The first is \emph{expansion}: creating an additional
slot in a STEM program and reserving it for a woman. The second is
\emph{replacement}: holding capacity fixed and admitting a woman in place
of the marginal man. Both are standard tools of affirmative action.

The cascade framework developed in
Sections~\ref{sec:heterogeneity}--\ref{sec:replacement} provides a
unified analysis of both policy instruments. The expansion effect is the
gender-conditional policy effect $T_k^{|f}$; the replacement effect is
$R_k^{f \leftarrow m} = T_k^{|f} - T_k^{|m}$. Each can be decomposed
by beneficiary: how much of the effect accrues to women's STEM degrees
versus men's? The cascade matters because the woman who enters and the
man who is displaced vacate different fallback programs, triggering
different downstream reallocations through the mixed-gender admission
queues. A policy that appears to move one woman into STEM may, through
the cascade, also move men into or out of STEM at other points in the
system.

We apply this decomposition to Swedish admissions data. Among applicants to STEM programs in the pivotal groups, approximately 57.7 percent ultimately obtain a STEM degree, while only 27.2 percent of applicants are female (see Table \ref{tab:summary_fields} in Appendix \ref{sec:appendix}). 

We classify STEM programs into three tiers based on admission selectivity: competitive STEM (top quartile), mid-tier STEM, and non-STEM. The outcome is a binary indicator for holding a STEM degree in 2022. This three-tier classification allows cascade effects to operate \emph{within} STEM, which would be obscured by a single STEM indicator. For example, admitting a woman to a competitive STEM program may free up a spot in a mid-tier STEM program, which is then filled from a predominantly male applicant pool.


\begin{figure}[htbp]
\centering

\vspace{0.5em}

\newcommand{\cell}[6]{%
    \pgfmathsetmacro{\abscoef}{abs(#3)}%
    \pgfmathsetmacro{\pct}{min(\abscoef / 0.55 * 65, 65)}%
    \pgfmathtruncatemacro{\pctint}{\pct}%
    \pgfmathtruncatemacro{\ispos}{#3 > 0 ? 1 : 0}%
    \ifnum\ispos=1
        \fill[red!\pctint!white] (#1, #2) rectangle ++(1, 1);
    \else
        \fill[blue!\pctint!white] (#1, #2) rectangle ++(1, 1);
    \fi
    \draw[gray!40] (#1, #2) rectangle ++(1, 1);
    \ifnum#5=1
        \node[anchor=center, font=\scriptsize\bfseries] at (#1+0.5, #2+0.5) {%
            \begin{tabular}{c} #6 \end{tabular}};
    \else
        \node[anchor=center, font=\scriptsize, text=gray!50] at (#1+0.5, #2+0.5) {%
            \begin{tabular}{c} #6 \end{tabular}};
    \fi
}

\begin{minipage}[t]{0.47\textwidth}
\centering
\textbf{(a) Women}

\vspace{0.3em}

\begin{tikzpicture}[
    x=2.1cm, y=-2.1cm,
    every node/.style={font=\scriptsize}
]

\input{stem_heatmap_cells_female.tex}

\node[anchor=south, font=\scriptsize\bfseries, align=center] at (0.5, -0.12) {Non-\\STEM};
\node[anchor=south, font=\scriptsize\bfseries, align=center] at (1.5, -0.12) {Mid-\\tier};
\node[anchor=south, font=\scriptsize\bfseries, align=center] at (2.5, -0.12) {Comp.\\STEM};

\node[anchor=east, font=\scriptsize\bfseries, align=right] at (-0.12, 0.5) {Non-STEM};
\node[anchor=east, font=\scriptsize\bfseries, align=right] at (-0.12, 1.5) {Mid-tier};
\node[anchor=east, font=\scriptsize\bfseries, align=right] at (-0.12, 2.5) {Competitive};

\end{tikzpicture}
\end{minipage}%
\hfill
\begin{minipage}[t]{0.47\textwidth}
\centering
\textbf{(b) Men}

\vspace{0.3em}

\begin{tikzpicture}[
    x=2.1cm, y=-2.1cm,
    every node/.style={font=\scriptsize}
]

\input{stem_heatmap_cells_male.tex}

\node[anchor=south, font=\scriptsize\bfseries, align=center] at (0.5, -0.12) {Non-\\STEM};
\node[anchor=south, font=\scriptsize\bfseries, align=center] at (1.5, -0.12) {Mid-\\tier};
\node[anchor=south, font=\scriptsize\bfseries, align=center] at (2.5, -0.12) {Comp.\\STEM};

\node[anchor=east, font=\scriptsize\bfseries, align=right] at (-0.12, 0.5) {Non-STEM};
\node[anchor=east, font=\scriptsize\bfseries, align=right] at (-0.12, 1.5) {Mid-tier};
\node[anchor=east, font=\scriptsize\bfseries, align=right] at (-0.12, 2.5) {Competitive};

\end{tikzpicture}
\end{minipage}

\caption{\footnotesize First-stage matrix by gender: STEM tiers. Each cell reports the first-stage
coefficient and $t$-statistic (in parentheses) from regressing
enrollment in the row tier on the column tier's lottery instrument,
estimated separately for women (panel~a) and men (panel~b). Red
shading indicates positive coefficients (own effects on the diagonal);
blue indicates negative (cross-program substitution). Bold entries are
significant at the 5\% level. All regressions include tier-of-application dummies,
age, year fixed effects, and application priority as controls. Standard
errors clustered by pivotal group.}
\label{fig:stem_heatmap_gender}

\end{figure}

Figure~\ref{fig:stem_heatmap_gender} displays the gender-specific first-stage matrices. The diagonal elements confirm strong compliance for both genders, with women showing somewhat larger own effects in competitive STEM (0.549 vs.\ 0.443) and mid-tier STEM (0.394 vs.\ 0.359), reflecting that women in pivotal groups for STEM programs are more strongly affected by the lottery, which is consistent with men having more fallback options within STEM. The off-diagonal elements reveal the substitution patterns that drive the cascade. The key asymmetry is in the competitive STEM column: when women win the competitive STEM lottery, they vacate mid-tier STEM at a rate of 0.134; about 24\% of the own effect. For men, the corresponding rate is 0.229, or 52\% of the own effect. Men admitted to competitive STEM are thus twice as likely as women to have come from mid-tier STEM rather than from outside STEM. Conversely, women admitted to competitive STEM vacate non-STEM slots at a higher rate than men (0.088 vs.\ 0.068), indicating that a larger share of marginal women in competitive STEM would otherwise have left STEM entirely. This asymmetry is the mechanism behind the gendered cascade: when a woman enters competitive STEM, the mid-tier vacancy she creates is smaller, but the direct effect on the marginal women is bigger. When a man enters, the direct effect on the marginal man is smaller, but the cascade is bigger since his admission frees more slots in mid-tier STEM. 

Table~\ref{tab:stem_main} reports the full-sample cascade decomposition.
Expanding competitive STEM by one slot generates 0.215 STEM degrees in
total, of which the cascade accounts for 0.087 ($p < 0.01$)---roughly
40\% of the total effect. The cascade is driven by substitution within
the STEM hierarchy: the marginal competitive-STEM admit vacates a
mid-tier slot, which is refilled from mid-tier's queue. For mid-tier
STEM, the cascade is smaller (0.026, $p < 0.05$) because the marginal
admit comes predominantly from outside STEM. The non-STEM row confirms
the logic from the other direction: the Wald ratio is significantly
negative ($-0.054$), but the positive cascade (0.034) partially offsets
the loss as vacated STEM slots are refilled.

\begin{table}[htbp]\centering
\caption{Expanding STEM tiers: effect on STEM degree attainment}
\label{tab:stem_main}
\begin{threeparttable}
\begin{tabular}{lccc}
\hline
                    & $\beta_k$ & $W_k$ & Cascade \\
                    & (1) & (2) & (3) \\
\hline
                    &\multicolumn{1}{c}{(1)}         &\multicolumn{1}{c}{(2)}         &\multicolumn{1}{c}{(3)}         \\
[1em]
Non-STEM            &     -0.0200         &     -0.0537\sym{***}&      0.0337\sym{***}\\
                    &    (0.0160)         &    (0.0148)         &   (0.00871)         \\
[1em]
Mid-tier STEM       &       0.219\sym{***}&       0.192\sym{***}&      0.0263\sym{*}  \\
                    &    (0.0524)         &    (0.0476)         &    (0.0137)         \\
[1em]
Competitive STEM    &       0.215\sym{***}&       0.128\sym{***}&      0.0870\sym{***}\\
                    &    (0.0479)         &    (0.0390)         &    (0.0272)         \\
[1em]
Observations        &      149579         &           .         &           .         \\
KP F-statistic      &       93.03         &           .         &           .         \\

\\
\hline
\end{tabular}
\begin{tablenotes}[flushleft]
\footnotesize
\item \textit{Notes:} Outcome is a binary indicator for holding a STEM
degree. Treatments are admission to one of three tiers defined by
programme competitiveness. Column~(1): full 2SLS estimate $\beta_k =
T_k$. Column~(2): single-instrument Wald ratio $W_k$. Column~(3):
cascade $T_k - W_k$. Standard errors for columns~(1)--(2) are clustered
by pivotal group; column~(3) uses bootstrap standard errors (1000
replications). \sym{*} $p<0.10$, \sym{**} $p<0.05$, \sym{***} $p<0.01$.
\end{tablenotes}
\end{threeparttable}
\end{table}

Table~\ref{tab:admit_gender} applies the gender-conditional formula of
equation~\eqref{eq:conditional_het} with STEM degree as the outcome,
decomposed into women's and men's STEM degrees. The headline
result is in the competitive STEM row in Panel A: admitting one more woman
generates 0.250 total STEM degrees, of which 0.197 are women's and
0.053 are men's ($p < 0.01$), approximately one man's STEM degree for
every four women's. The effect on men is entirely a cascade effect: the
woman admitted to competitive STEM vacates a mid-tier slot that is
refilled from a predominantly male queue.
 
The mid-tier row reveals a different pattern. Admitting one more woman
to mid-tier STEM generates essentially zero additional women's STEM
degrees ($-0.022$, insignificant): the marginal woman does not increase her chance of ending up with a STEM degree when admitted to mid-tier STEM. Mid-tier
STEM is thus not the binding constraint for women's STEM attainment. The
effect on men is positive (0.028) but imprecisely estimated. 
 
The non-STEM row is instructive. Admitting one more woman to a non-STEM
program has no net effect on total STEM degrees (0.001), but the
decomposition shows why: it reduces women's STEM degrees by 0.016
($p < 0.01$) and increases men's by the same amount ($p < 0.01$). Some
women drawn into non-STEM would otherwise have been in STEM; these
vacated STEM slots are mostly filled by men. The policy is almost exactly
zero-sum for total STEM, but it redistributes degrees from women to men.

\begin{table}[htbp]\centering
\caption{Expanding field capacity on STEM degree attainment}
\label{tab:admit_gender}
\begin{threeparttable}
\renewcommand{\arraystretch}{1.15}
\begin{tabular}{lcccc}
\toprule
                    & Total & Direct & $\to$ Women's & $\to$ Men's \\
                    & $T_k^{|g}$ & $W_k^{|g}$ & $T_k^{|g,f}$ & $T_k^{|g,m}$ \\
                    & (1) & (2) & (3) & (4) \\
\midrule
\multicolumn{5}{c}{\textit{A. Admit one more woman}} \\[0.3em]
Non-STEM            &    0.000666         &     -0.0168         &     -0.0155         &      0.0162\sym{**} \\
                    &    (0.0166)         &    (0.0165)         &    (0.0168)         &   (0.00662)         \\
Mid-tier STEM       &     0.00625         &     -0.0290         &     -0.0215         &      0.0278\sym{*}  \\
                    &    (0.0900)         &    (0.0917)         &    (0.0932)         &    (0.0146)         \\
Competitive STEM    &       0.250\sym{***}&       0.197\sym{***}&       0.197\sym{***}&      0.0529\sym{**} \\
                    &    (0.0686)         &    (0.0684)         &    (0.0703)         &    (0.0211)         \\
Observations        &       88545         &           .         &           .         &           .         \\
KP F-statistic      &       32.93         &           .         &           .         &           .         \\

\\[0.5em]
\multicolumn{5}{c}{\textit{B. Admit one more man}} \\[0.3em]
Non-STEM            &     -0.0558\sym{*}  &      -0.119\sym{***}&     0.00707         &     -0.0629\sym{*}  \\
                    &    (0.0327)         &    (0.0313)         &   (0.00888)         &    (0.0327)         \\
Mid-tier STEM       &       0.312\sym{***}&       0.286\sym{***}&     0.00497         &       0.307\sym{***}\\
                    &    (0.0636)         &    (0.0592)         &   (0.00704)         &    (0.0636)         \\
Competitive STEM    &       0.206\sym{***}&      0.0961\sym{*}  &    0.000569         &       0.206\sym{***}\\
                    &    (0.0578)         &    (0.0495)         &    (0.0189)         &    (0.0581)         \\
Observations        &       61034         &           .         &           .         &           .         \\
KP F-statistic      &       52.47         &           .         &           .         &           .         \\

\\
\bottomrule
\end{tabular}
\begin{tablenotes}[flushleft]
\footnotesize
\item \textit{Notes:} Outcome is a binary indicator for STEM degree.
Column~(1) reports the total effect of admitting one more student of gender $g$
to tier~$k$. Column~(2) reports the direct effect on the marginal admitted student.
Columns~(3)--(4) decompose the total effect into effects on women's and men's STEM
degrees. Bootstrap standard errors clustered by pivotal group.
\sym{*} $p<0.10$, \sym{**} $p<0.05$, \sym{***} $p<0.01$.
\end{tablenotes}
\end{threeparttable}
\end{table}

Panel B of Table~\ref{tab:admit_gender} reports the symmetric analysis for men. The contrast with Panel A is sharp. Admitting one more
man to competitive STEM generates 0.205 STEM degrees, virtually all of
which are men's (0.205) and essentially none women's (0.001). Unlike the
female case, the male cascade does not benefit the other gender: men at
the competitive STEM margin vacate programs whose queues do not contain
women who would gain STEM degrees. The asymmetry traces directly to the
first-stage matrices in Figure~\ref{fig:stem_heatmap_gender}: men
admitted to competitive STEM vacate mid-tier STEM at a much higher rate
than women do (0.229 vs.\ 0.134), but the mid-tier refill draws
predominantly from men, so the cascade circulates within the male STEM
pipeline without pulling women in.

Table~\ref{tab:replace} reports the replacement effect
$R_k^{f \leftarrow m} = T_k^{|f} - T_k^{|m}$: the net consequence of
replacing the marginal man with a woman, holding capacity fixed. The
point estimates follow mechanically from Table \ref{tab:admit_gender}; the contribution here is inference on the difference.
 
At competitive STEM, the net effect on total STEM degrees is small and
insignificant (0.043): the quota barely changes total STEM production.
But the decomposition reveals that this near-zero masks a significant
redistribution: the quota generates 0.196 additional women's STEM
degrees ($p < 0.01$) and destroys 0.153 men's ($p < 0.01$). The quota
redistributes rather than creates.
 
The redistribution is not one-for-one. Each woman who replaces a man
generates 0.196 women's STEM degrees but eliminates only 0.153 men's
--- the net gain of 0.043 reflects the asymmetry in the two cascades
documented above. Because the woman's fallback is more likely to be
outside STEM while the man's is mid-tier STEM, removing the man
destroys fewer downstream STEM degrees than admitting the woman creates.
The quota is mildly STEM-expanding, but the effect is too small to
distinguish from zero.

\begin{table}[htbp]\centering
\caption{Replace the marginal man with a woman: net effect on STEM degrees}
\label{tab:replace}
\begin{threeparttable}
\begin{tabular}{lcccc}
\hline
                    & Net total & Net direct & $\to$ Women's & $\to$ Men's \\
                    & $\Delta T_k$ & $\Delta W_k$ & $\Delta T_k^{f}$ & $\Delta T_k^{m}$ \\
                    & (1) & (2) & (3) & (4) \\
\hline
Non-STEM            &      0.0565         &       0.102\sym{***}&     -0.0226         &      0.0791\sym{**} \\
                    &    (0.0356)         &    (0.0361)         &    (0.0171)         &    (0.0320)         \\
Mid-tier STEM       &      -0.306\sym{***}&      -0.315\sym{***}&     -0.0265         &      -0.279\sym{***}\\
                    &     (0.108)         &     (0.111)         &    (0.0904)         &    (0.0604)         \\
Competitive STEM    &      0.0433         &       0.101         &       0.196\sym{***}&      -0.153\sym{***}\\
                    &    (0.0865)         &    (0.0887)         &    (0.0677)         &    (0.0526)         \\

\\
\hline
\end{tabular}
\begin{tablenotes}[flushleft]
\footnotesize
\item \textit{Notes:} Each cell is $T_k^{|f} - T_k^{|m}$: the net
effect of replacing the marginal man with a woman in tier~$k$, holding
capacity fixed. Column~(1): net effect on total STEM degrees.
Column~(2): net direct effect. Columns~(3)--(4): net effect on women's
and men's STEM degrees. A positive value in column~(3) and negative
in column~(4) indicates redistribution from men to women. Bootstrap
standard errors (1000 replications, clustered by pivotal group).
\sym{*} $p<0.10$, \sym{**} $p<0.05$, \sym{***} $p<0.01$.
\end{tablenotes}
\end{threeparttable}
\end{table}

These results illustrate a general feature of affirmative action in
capacity-constrained systems: the policy does not operate in a vacuum.
Whether the instrument is expansion or replacement, the cascade
propagates the intervention through the system, generating
distributional consequences that extend beyond the individuals directly
targeted. In our setting, expanding competitive STEM toward women
generates one man's STEM degree for every four women's; replacing the
marginal man with a woman redistributes degrees across genders but
barely changes the total. The cascade framework makes these downstream
effects visible and measurable using only standard IV output. The same logic applies to any rationed setting where slots are
reallocated across types---including corporate board quotas
\citep{matsa2013, ahern2012} and race-based university admissions
policies, where banning or introducing affirmative action triggers
cascading reallocations across tiers of the system
\citep{bleemer2022}.

\section{Conclusion}
\label{sec:conclusion}

The prevailing view in the econometric literature is that 2SLS with
multiple treatments yields coefficients that are difficult to interpret
under heterogeneous treatment effects. We show that in
capacity-constrained allocation systems---a class that includes
university admissions, school choice, medical residency matching, public
housing, and other rationed settings---the 2SLS coefficient $\beta_k$
has a direct policy interpretation: the total societal effect of
expanding treatment $k$ by one slot, including all cascading
reallocations through the system. The result is an algebraic identity
that holds for any first-stage matrix, requires only instrument relevance
and that the allocation mechanism transmits supply expansions through the
instruments, and imposes no restrictions on treatment effect heterogeneity
or individual choice behavior. The identity applies not only to
queue-based systems but to any allocation mechanism operating over
goods with fixed supply, including competitive markets with price
instruments.

For applied researchers working in capacity-constrained settings, the
practical implication is straightforward. Multi-treatment 2SLS is not
merely a convenience for combining instruments, it is the estimator
that answers the policy question of whether to expand capacity. The
cascade decomposition into own effects and downstream reallocation
effects requires no additional estimation beyond what is already standard
practice: the difference between the multi-treatment 2SLS coefficient and
the single-instrument Wald ratio recovers the cascade directly.

The result has limitations. Assumption~\ref{ass:margin} requires that
the instrument and the policy operate on the same margin: a condition
satisfied by construction in centralized systems but potentially violated
in decentralized settings where vacancies are filled through different
channels than those generating the instrument variation. The fixed-supply
condition is essential: if competing programs can expand endogenously in
response to the demand pressure created by the cascade, the identity
breaks down. In practice, most publicly funded systems enforce fixed
supply through administrative budget separation, regulatory constraints,
or policy choice, but the assumption should be assessed case by case.
Finally, the cascade identity characterizes a system-level policy effect
that is in general difficult to decompose into individual-level treatment
effects. Under homogeneous effects, the system-level and individual-level
parameters coincide; under heterogeneity, they diverge, requiring the
researcher to be clear about which parameter is being targeted.

Our empirical applications demonstrates the cascade framework using
Swedish university admissions. The results show
that the cascade correction is quantitatively important: for competitive
fields like business and social science, the entire policy effect of
expansion operates mainly through the downstream reallocation of displaced
students rather than through the direct effect on the marginal entrant.
This finding reframes the long-standing debate on whether economics
education erodes prosocial values: the apparent prosociality gap between
economics students and others is driven primarily by the prosocial
effects of the fields that economics students would otherwise have
attended, not by any negative effect of economics itself. The second application demonstrates that gender-targeted STEM admission policies generate distributional consequences that extend well beyond the individuals directly affected: approximately one-fifth of the STEM-degree effect of admitting a woman to competitive STEM accrues to men through the cascade.

\newpage
\bibliographystyle{econ}
\bibliography{cascade}

\newpage

\appendix

\section{Application appendix}
\label{sec:appendix}
\subsection{Heatmap with more granular programs}
\label{sec:heatmap2}
%
%
%

\begin{figure}[!ht]
\centering
\resizebox{\textwidth}{!}{%
\begin{tikzpicture}[x=0.42cm, y=0.42cm]

\def\maxabs{0.35}
\def\opSig{1}
\def\opNS{0.15}

%
\newcommand{\cell}[6]{%
  \pgfmathsetmacro{\cellcoef}{#3}%
  \pgfmathsetmacro{\intensity}{sqrt(min(abs(\cellcoef)/\maxabs,1))*100}%
  \ifdim \cellcoef pt > 0.00001pt
    \def\basecolor{red}%
  \else
    \ifdim \cellcoef pt < -0.00001pt
      \def\basecolor{blue}%
    \else
      \def\basecolor{white}%
    \fi
  \fi
  \ifnum #5=1
    \def\myop{\opSig}%
  \else
    \def\myop{\opNS}%
  \fi
  \filldraw[draw=black!10,
            fill=\basecolor!\intensity!white,
            fill opacity=\myop]
      (#1, 46-#2) rectangle ++(1,1);
}

\foreach \j/\lab in {
   0/{Tech. intro.},
   1/{Educ.\ (gen.)},
   2/{Pedagogy},
   3/{Early child.\ tchg},
   4/{Primary tchg},
   5/{Secondary tchg},
   6/{Vocational tchg},
   7/{Educ.\ (other)},
   8/{Arts \& Media},
   9/{Humanities},
  10/{Soc.\ \& Behav.\ Sci.},
  11/{Psychology},
  12/{Sociology},
  13/{Pol.\ Science},
  14/{Economics},
  15/{Journalism},
  16/{Business},
  17/{Law},
  18/{Biology \& Env.},
  19/{Physics \& Chem.},
  20/{Mathematics},
  21/{ICT},
  22/{Engineering (gen.)},
  23/{Mech.\ Engineering},
  24/{Energy \& Electr.},
  25/{Electronics \& IT},
  26/{Chem.\ \& Biotech.},
  27/{Vehicle Eng.},
  28/{Industrial Econ.},
  29/{Env.\ Engineering},
  30/{Materials},
  31/{Civil Engineering},
  32/{Agriculture},
  33/{Veterinary},
  34/{Healthcare (gen.)},
  35/{Medicine},
  36/{Nursing},
  37/{Dentistry},
  38/{Health Technology},
  39/{Therapy \& Rehab.},
  40/{Pharmacy},
  41/{Healthcare (other)},
  42/{Social Work},
  43/{Personal Services},
  44/{Transport},
  45/{Occup.\ Safety},
  46/{Security Services}
}{
  \node[rotate=60, anchor=west, font=\scriptsize] at (\j+0.5, 47.3) {\lab};
}

\foreach \i/\lab in {
   0/{Tech. intro.},
   1/{Educ.\ (gen.)},
   2/{Pedagogy},
   3/{Early child.\ tchg},
   4/{Primary tchg},
   5/{Secondary tchg},
   6/{Vocational tchg},
   7/{Educ.\ (other)},
   8/{Arts \& Media},
   9/{Humanities},
  10/{Soc.\ \& Behav.\ Sci.},
  11/{Psychology},
  12/{Sociology},
  13/{Pol.\ Science},
  14/{Economics},
  15/{Journalism},
  16/{Business},
  17/{Law},
  18/{Biology \& Env.},
  19/{Physics \& Chem.},
  20/{Mathematics},
  21/{ICT},
  22/{Engineering (gen.)},
  23/{Mech.\ Engineering},
  24/{Energy \& Electr.},
  25/{Electronics \& IT},
  26/{Chem.\ \& Biotech.},
  27/{Vehicle Eng.},
  28/{Industrial Econ.},
  29/{Env.\ Engineering},
  30/{Materials},
  31/{Civil Engineering},
  32/{Agriculture},
  33/{Veterinary},
  34/{Healthcare (gen.)},
  35/{Medicine},
  36/{Nursing},
  37/{Dentistry},
  38/{Health Technology},
  39/{Therapy \& Rehab.},
  40/{Pharmacy},
  41/{Healthcare (other)},
  42/{Social Work},
  43/{Personal Services},
  44/{Transport},
  45/{Occup.\ Safety},
  46/{Security Services}
}{
  \node[anchor=east, font=\scriptsize] at (-0.3, 46-\i+0.5) {\lab};
}

\node[font=\small\bfseries] at (23.5, 56)
  {Lottery instrument (column) $\to$ first-stage outcome (row)};

\input{heatmap_cells_programs}

\draw[black!30] (0,0) rectangle (47,47);

\end{tikzpicture}
}
\caption{First-stage matrix for more granular field definitions. Each cell depicts the sign and strength of the lottery instrument on being admitted to each program
(first-stage coefficients). Red~=~positive, blue~=~negative.
Faded~=~not significant ($|t|<1.96$).}
\label{fig:firststage_programs}
\end{figure}

\newpage

\subsection{Aggregation across blocks of programs}
\label{sec:aggregation}

Programs are often grouped into coarser \emph{blocks}---for example fields of study---either
because program-level instruments are unavailable or because the research question concerns
broader categories. In practice, grouping is also necessary for inference, computation, and
interpretation: with thousands of programs, estimating a fully disaggregated system is
infeasible, and even where feasible, the resulting program-level estimates are too noisy to
be informative. The standard approach is to assume
treatment homogeneity within blocks---that is, $\beta_m = \beta_B$ for all $m \in
\mathcal{J}_B$---under which the block coefficient inherits the T-interpretation directly:
it equals the societal effect of expanding the block by one slot, since all programs within
the block have the same policy effect by assumption.

Let $A_{im}$ denote enrollment in program $m$ for $m=1,\dots,M$. For a block $B$
consisting of programs $\mathcal{J}_B \subseteq \{1,\dots,M\}$ define the block treatment
as
\begin{equation}
A_{iB}=\sum_{m\in\mathcal{J}_B}A_{im},
\end{equation}
which is binary when programs are mutually exclusive. The most granular specification is
\begin{equation}
Y_i=\beta_0+\sum_{m=1}^M\beta_m A_{im}+u_i,
\end{equation}
while the block-level specification replaces the program indicators with the block treatment,
\begin{equation}
Y_i=\beta_0+\beta_B A_{iB}+v_i.
\end{equation}
When both specifications are estimated on the same sample with the same controls, the block
coefficient can be written as a weighted average of the underlying program coefficients,
\begin{equation}
\beta_B=\sum_{m\in\mathcal{J}_B}\omega_{m|B}\beta_m,
\qquad
\omega_{m|B}=
\frac{\operatorname{Cov}(\tilde Z_B,\tilde A_{im})}
{\sum_{j\in\mathcal{J}_B}\operatorname{Cov}(\tilde Z_B,\tilde A_{ij})},
\label{eq:blockweights}
\end{equation}
where tildes denote residuals after partialling out included exogenous variables. Under
treatment homogeneity within blocks, $\beta_B = \beta_m$ for all $m \in \mathcal{J}_B$,
and the weighted average collapses to the common value regardless of the weights. More
generally, without homogeneity, the block coefficient aggregates program-specific policy
effects with weights determined by how strongly the block-level instrument shifts enrollment
into each constituent program.

Without program-specific instruments---when a single pooled instrument is
used---the weights in equation~\eqref{eq:blockweights} can be negative.
This occurs when monotonicity fails: if applicants winning a competitive
program vacate a secondary program within the same block at a higher rate
than the instrument directly fills it, the block coefficient remains a
well-defined linear combination of the underlying policy effects but need
not lie between them, and its policy interpretation becomes difficult to
characterize.

A natural choice of block instrument is the vector of program-specific
instruments $Z_{im} = \text{luck}_i \times \mathbf{1}(\text{applied to
}m)$, one for each program in the block. Using these interacted instruments
and including program-of-application dummies $\mathbf{1}(\text{applied to
}m)$ as exogenous controls eliminates cross-program contamination and
ensures strictly positive weights as in equation~\eqref{eq:diagweights}. The block IV estimator is
\begin{equation}
\beta_B
=
\frac{\operatorname{Cov}(\tilde Z_B^{\mathrm{eff}},\tilde Y)}
{\operatorname{Cov}(\tilde Z_B^{\mathrm{eff}},\tilde A_{iB})},
\end{equation}
where $\tilde Z_B^{\mathrm{eff}}$ is the fitted value from projecting $A_{iB}$ onto the
interacted instruments and controls. Because $Z_{im}$ is zero for all applicants outside
program $m$'s pivotal group, the cross-program covariances $\operatorname{Cov}(\tilde
Z_{im}, \tilde A_{ij})$ are zero for $j\neq m$ after partialling out the
program-of-application dummies. The weights in equation~\eqref{eq:blockweights} therefore
reduce to
\begin{equation}
\omega_{m|B}=\frac{\pi_{mm}}{\sum_{j\in\mathcal{J}_B}\pi_{jj}},
\label{eq:diagweights}
\end{equation}
where $\pi_{mm}=\operatorname{Cov}(\tilde Z_{im},\tilde A_{im})$ is the direct first-stage
effect of program $m$'s own instrument on enrollment in program $m$. These weights are
strictly positive and sum to one, so $\beta_B$ is a proper weighted average of the
program-specific policy effects $\beta_m$.

When the block covers all programs in the system, the pooled IV estimator with
program-specific instruments has a particularly clean interpretation. The first stage
of each interacted instrument $Z_{im}$ on total enrollment $A_{iB}$ is zero for
inframarginal students, i.e. those who would enroll somewhere regardless of any individual
lottery outcome, and positive only for students on the extensive margin of higher
education. The block coefficient therefore measures the societal effect of expanding
university capacity for students who would not otherwise attend, with weights proportional
to how many such students each program draws from outside higher education altogether.
A program that attracts only students who would otherwise have enrolled elsewhere
receives zero weight, even if its program-specific policy effect $\beta_k$ is large:
it contributes to the cascade within higher education but not to the extensive margin
that this estimator identifies.

Table~\ref{tab:aggregation} reports this estimate under three instrument sets of
increasing granularity. The coefficient on any higher education enrollment is
stable across columns---0.044, 0.040, and 0.038---and remains highly significant
throughout, despite the Kleibergen-Paap F-statistic falling from 473 with a single
pooled instrument to 16 with 47 subfield instruments. This stability is informative
beyond mere robustness: under treatment homogeneity, all three specifications
identify the same parameter, and disagreement across instrument sets would signal
either instrument invalidity or heterogeneous treatment effects at different margins
of the distribution; the overidentification logic of \citet{angristpischke2009} and the
monotonicity diagnostics of \citet{heckman2005}. The consistency of the estimates
across instrument designs with very different first-stage structures is therefore
evidence that monotonicity holds within the full block of higher education and that
the estimated effect is not an artifact of a particular instrument's identifying
margin. Expanding higher education capacity along the extensive margin increases
charitable giving by approximately 0.04, an effect that is robust to the choice of
instrument granularity.

\begin{table}[htbp]\centering
\caption{Aggregation: effect of any higher education}
\label{tab:aggregation}
\begin{threeparttable}
\begin{tabular}{lccc}
\hline
                    &\multicolumn{1}{c}{(1)}         &\multicolumn{1}{c}{(2)}         &\multicolumn{1}{c}{(3)}         \\
\hline
Any higher education&      0.0443\sym{***}&      0.0403\sym{***}&      0.0377\sym{***}\\
                    &    (0.0106)         &    (0.0103)         &   (0.00979)         \\
[1em]
Prior giving        &       0.461\sym{***}&       0.461\sym{***}&       0.456\sym{***}\\
                    &   (0.00894)         &   (0.00895)         &   (0.00892)         \\
[1em]
Female              &      0.0328\sym{***}&      0.0327\sym{***}&      0.0334\sym{***}\\
                    &   (0.00105)         &   (0.00105)         &   (0.00107)         \\
[1em]
Age                 &     0.00305\sym{***}&     0.00302\sym{***}&     0.00271\sym{***}\\
                    &  (0.000125)         &  (0.000122)         &  (0.000123)         \\
\hline
Instruments         &Single (luck)         &    7 fields         &47 subfields         \\
Field FE            &         Yes         &         Yes         &         Yes         \\
Year FE             &         Yes         &         Yes         &         Yes         \\
KP F-statistic      &     473.948         &      98.270         &      16.559         \\
Observations        &      149579         &      149579         &      149579         \\
\\
\hline
\end{tabular}
\begin{tablenotes}[flushleft]
\footnotesize
\item \textit{Notes:} The dependent variable is an indicator for charitable giving
in the years following the admission decision. The endogenous variable is enrollment
in any higher education program. Column~(1) uses a single pooled lottery instrument;
column~(2) uses seven field-specific lottery instruments interacted with
field-of-application dummies; column~(3) uses 47 subfield-specific lottery
instruments interacted with subfield-of-application dummies (see
Appendix~\ref{sec:heatmap2} for the first-stage matrix). All specifications include
field-of-application dummies, year fixed effects, prior giving, gender, age, and
application priority as controls. Standard errors clustered by pivotal group. KP
F-statistic is the Kleibergen-Paap rk Wald F-statistic for weak identification.
\end{tablenotes}
\end{threeparttable}
\end{table}

\newpage 

\subsection{Descriptive statistics}
Our sample comprises individuals at the program admission threshold (the pivotal group). They are typically young and predominantly female, except in business and STEM fields. Application priorities range from 1 to 20, with lower values indicating higher priority; the mean priority is approximately 2–3. See Table \ref{tab:summary_fields}.

\begin{table}[htbp]\centering
\caption{Summary statistics by field of study}
\label{tab:summary_fields}
\begin{threeparttable}
\renewcommand{\arraystretch}{1.2}
\begin{tabular}{lccccccc}
\toprule
            & Business & Social science & Teaching & Medicine & Health & STEM & Other \\
\midrule
Female      &       0.411&       0.636&       0.758&       0.602&       0.803&       0.272&       0.685\\
Charitable giving&       0.082&       0.091&       0.164&       0.127&       0.116&       0.065&       0.094\\
STEM degree &       0.096&       0.056&       0.106&       0.058&       0.029&       0.577&       0.123\\
App. priority&       3.141&       3.458&       2.292&       2.827&       3.482&       2.786&       2.714\\
Age         &      23.418&      25.293&      34.858&      26.319&      27.508&      22.717&      25.447\\
Observations&       19882&       47011&       11102&       21169&       26867&       20358&        3190\\

\\[-0.8em]
\bottomrule
\end{tabular}
\begin{tablenotes}[flushleft]
\footnotesize
\item \textit{Notes:} The table reports mean values for each variable by field of study. Charitable giving is an indicator equal to one if the individual donated to charity in 2021 (last year of data). STEM degree is an indicator equal to one if the individual holds a completed university degree in a STEM field in 2022 (last year of data). Observations correspond to the number of individuals in each field. Age, female and application priority are observed at the year of application.
\end{tablenotes}

\end{threeparttable}
\end{table}

The higher mean age of teaching applicants ($\approx$34 vs. 25 in other fields) is largely driven by programs targeting mid-career professionals, such as \emph{Kompletterande pedagogisk utbildning} (KPU) and \emph{Specialpedagog}. These programs account for about 4\% of the sample and have a mean age of 40; excluding them lowers the mean age of teaching applicants to $\approx$29, closer to the overall average. Demand may also reflect the 2011 \emph{lärarlegitimation} reform, which introduced certification requirements for permanent teaching positions.

\subsection{Balance and Robustness of Application 1}
\label{sec:balanceandrobustness}
The balance test in Table~\ref{tab:balance} verifies that the instrument is orthogonal to individual-level predetermined characteristics. We regress each control variable on $L_i$ without fixed effects, which is the appropriate specification given that the instrument is already normalized within groups. All four coefficients are small and statistically indistinguishable from zero, and the joint F-statistic of 0.63 ($p = 0.64$) confirms that the lottery is as good as randomly assigned.

\begin{table}[htbp]\centering
\caption{Balance test}
\label{tab:balance}
\begin{threeparttable}
\begin{tabular}{p{6cm}c}
\hline
& Coefficient on $L_i$ \\
\hline
\\
Prior giving        &    0.000115         \\
                    &   (0.00117)         \\
\\
Female              &    -0.00317         \\
                    &   (0.00425)         \\
\\
Age                 &     -0.0480         \\
                    &    (0.0484)         \\
\\
Application priority&      0.0252         \\
                    &    (0.0257)         \\
\\
Joint F-statistic   &       0.630         \\
Joint p-value       &       0.641         \\
Observations        &      149579         \\
\\
\hline
\end{tabular}
\begin{tablenotes}[flushleft]
\footnotesize
\item \textit{Notes:} Each row reports the coefficient from a separate OLS regression of the row variable on the lottery instrument $L_i$, clustered by pivotal group. The joint F-statistic tests the null that all coefficients are simultaneously zero.
\end{tablenotes}
\end{threeparttable}
\end{table}

Table~\ref{tab:robustness} examines the sensitivity of the main estimates (first column of Table~\ref{tab:main}) to the inclusion of controls and fixed effects. Column~(1) includes field-of-application dummies only. Columns~(2) and~(3) add prior giving and then the full set of demographic controls. Column~(4) adds year fixed effects (this column coincides with our preferred model in the main results table \ref{tab:main}) and column~(5) absorbs pivotal group fixed effects via within-group demeaning. The field-level coefficients are stable throughout: point estimates shift by less than a percentage point across specifications and significance patterns are unchanged. The Kleibergen-Paap F-statistic remains close to 54-55 in all columns, confirming that the instrument strength is unaffected by the inclusion of controls---as expected given the orthogonality established in Table~\ref{tab:balance}. Among the control variables, prior giving is strongly persistent, women give more than men, giving increases with age, and applicants who ranked the pivotal program lower in their preference list give slightly less---a pattern consistent with weaker attachment to the field rather than a threat to identification.

\begin{table}[htbp]\centering
\caption{Robustness to control function}
\label{tab:robustness}
\begin{threeparttable}
\begin{tabular}{lccccc}
\hline
                    &\multicolumn{1}{c}{(1)}         &\multicolumn{1}{c}{(2)}         &\multicolumn{1}{c}{(3)}         &\multicolumn{1}{c}{(4)}         &\multicolumn{1}{c}{(5)}         \\
\hline
Business            &      0.0227         &      0.0197         &      0.0198         &      0.0196         &      0.0200         \\
                    &    (0.0152)         &    (0.0148)         &    (0.0149)         &    (0.0149)         &    (0.0149)         \\
[1em]
Social science      &      0.0344\sym{**} &      0.0377\sym{**} &      0.0367\sym{**} &      0.0368\sym{**} &      0.0371\sym{**} \\
                    &    (0.0163)         &    (0.0157)         &    (0.0156)         &    (0.0156)         &    (0.0156)         \\
[1em]
Teaching            &      0.0559\sym{**} &      0.0559\sym{**} &      0.0578\sym{***}&      0.0579\sym{***}&      0.0566\sym{**} \\
                    &    (0.0257)         &    (0.0222)         &    (0.0222)         &    (0.0221)         &    (0.0221)         \\
[1em]
Medicine            &      0.0692\sym{**} &      0.0684\sym{**} &      0.0765\sym{***}&      0.0767\sym{***}&      0.0741\sym{***}\\
                    &    (0.0271)         &    (0.0266)         &    (0.0266)         &    (0.0266)         &    (0.0266)         \\
[1em]
Health              &      0.0520\sym{**} &      0.0498\sym{**} &      0.0531\sym{***}&      0.0531\sym{***}&      0.0525\sym{***}\\
                    &    (0.0205)         &    (0.0201)         &    (0.0199)         &    (0.0199)         &    (0.0200)         \\
[1em]
STEM                &      0.0528\sym{***}&      0.0452\sym{***}&      0.0455\sym{***}&      0.0452\sym{***}&      0.0457\sym{***}\\
                    &    (0.0170)         &    (0.0166)         &    (0.0168)         &    (0.0168)         &    (0.0168)         \\
[1em]
Other               &     -0.0257         &     -0.0305         &     -0.0251         &     -0.0252         &     -0.0264         \\
                    &    (0.0232)         &    (0.0204)         &    (0.0202)         &    (0.0201)         &    (0.0202)         \\
[1em]
Prior giving        &                     &       0.482\sym{***}&       0.447\sym{***}&       0.461\sym{***}&       0.450\sym{***}\\
                    &                     &   (0.00926)         &   (0.00908)         &   (0.00898)         &   (0.00910)         \\
[1em]
Female              &                     &                     &      0.0320\sym{***}&      0.0312\sym{***}&      0.0317\sym{***}\\
                    &                     &                     &   (0.00229)         &   (0.00227)         &   (0.00177)         \\
[1em]
Age                 &                     &                     &     0.00279\sym{***}&     0.00295\sym{***}&     0.00168\sym{***}\\
                    &                     &                     &  (0.000151)         &  (0.000148)         &  (0.000143)         \\
[1em]
Application priority&                     &                     &    -0.00271\sym{***}&    -0.00258\sym{***}&    -0.00185\sym{***}\\
                    &                     &                     &  (0.000249)         &  (0.000239)         &  (0.000308)         \\
\hline
Pivotal group FE    &          No         &          No         &          No         &          No         &         Yes         \\
Year FE             &          No         &          No         &          No         &         Yes         &          No         \\
KP F-statistic      &      55.078         &      55.144         &      53.817         &      53.654         &      55.152         \\
Observations        &      149579         &      149579         &      149579         &      149579         &      149579         \\
\\
\hline
\end{tabular}
\begin{tablenotes}[flushleft]
\footnotesize
\item \textit{Notes:} Outcome variable is a binary indicator of charitable giving after the year of admittance. All specifications include field-of-application dummies. Standard errors clustered by pivotal group. Year fixed effects are subsumed by pivotal group fixed effects in column~(5).
\end{tablenotes}
\end{threeparttable}
\end{table}

\pagebreak
\section{Theory appendix}
\subsection{The cascade identity under general allocation mechanisms}
\label{app:general}

The cascade identity derived in Section~\ref{sec:model} is stated in the
language of ranked queues, but its logic does not depend on the allocation
mechanism. This appendix proves the result for an arbitrary mechanism,
clarifying what is essential and what is institutional detail.

\paragraph{Setup.}
There are $K$ goods with fixed total supply $\bar{Q}_1, \dots, \bar{Q}_K$
and an outside option with unrestricted supply. Individuals are allocated
quantities $Q_{ik}$ through some mechanism parameterized by variables
$Z = (Z_1, \dots, Z_K)$. In a queue system, $Z_k$ is the lottery rank or
cutoff for program~$k$. In a market, $Z_k$ is the price of good~$k$. In a
matching algorithm, $Z_k$ may be a priority score.
The mechanism is otherwise unrestricted.

The first-stage matrix has entries $\pi_{jk} = \partial\, E[Q_{ij}] /
\partial\, Z_k$, the reduced form is $\mathrm{RF}_k = \partial\, E[Y_i] /
\partial\, Z_k$, and the 2SLS coefficients satisfy $\bfRF = \Pi^T \bfbeta$
by construction.

\paragraph{The key restriction.}
We require two conditions:

\begin{description}
\item[Assumption 1$'$ (Relevance).] $\pi_{kk} \neq 0$ for all $k$.

\item[Assumption 2$'$ (Supply--instrument alignment).] When the supply of
good~$k$ increases by one unit, holding all other supplies fixed, the
mechanism restores equilibrium by adjusting $Z_k$. That is, $Z_k$ is the
variable through which a marginal supply expansion of good~$k$ is
transmitted to individual allocations.
\end{description}

Assumption~2$'$ is the general form of ranking-list refill. In a queue
system, adding a slot to program~$k$ shifts the admission cutoff---a change
in $Z_k$---and the ranking list determines who is affected. In a market,
adding a unit of good~$k$ lowers its price---again a change in $Z_k$---and
the demand system determines who is affected. In both cases, the first-stage
matrix $\Pi$ captures how individual allocations respond to the mechanism
variable, and the same matrix governs the response to a supply expansion.
The specific institution determines the \emph{magnitudes} of $\pi_{jk}$
and $\mathrm{RF}_k$, but the algebraic structure is the same.

\paragraph{Derivation.}
Let $T_k = dY / d\bar{Q}_k$ denote the total societal effect of increasing
the supply of good~$k$ by one unit, holding all other supplies fixed and
letting the mechanism adjust to restore equilibrium.

By Assumption~2$'$, the supply expansion operates through $Z_k$. The
direct effect per unit of new supply is $\mathrm{RF}_k / \pi_{kk}$: the
reduced-form outcome change per unit of $Z_k$, rescaled by how many units
of good~$k$ a one-unit change in $Z_k$ delivers.

But the change in $Z_k$ also shifts allocations of other goods: per unit
increase in $Z_k$, the total allocation to good~$j$ changes by $\pi_{jk}$.
Since the supply of good~$j$ is fixed, the mechanism must adjust $Z_j$ to
undo this change, restoring $\sum_i Q_{ij} = \bar{Q}_j$. This offsetting
adjustment in $Z_j$ is itself a marginal supply-neutral reallocation of
good~$j$ across individuals. Per unit of good~$k$ expanded, the magnitude
of the reallocation of good~$j$ is $|\pi_{jk}| / \pi_{kk}$.

When $\pi_{jk} < 0$ (the typical substitution case), the initial shock
reduces total demand for good~$j$, so the mechanism must push $Z_j$ in the
direction that increases good-$j$ allocations---exactly as if good~$j$'s
supply had expanded. The societal effect of this
reallocation is $T_j$ per unit. When $\pi_{jk} > 0$ (complementarity), the
mechanism must contract good-$j$ allocations, with effect $-T_j$ per unit.
In both cases, the contribution to $T_k$ is $(-\pi_{jk}/\pi_{kk}) \cdot T_j$.

Summing:
\begin{equation}
T_k = \frac{\mathrm{RF}_k}{\pi_{kk}} + \sum_{j \neq k}
\frac{-\pi_{jk}}{\pi_{kk}} \, T_j, \qquad k = 1, \dots, K.
\end{equation}
This is identical to equation~\eqref{eq:cascade}. The remainder of the proof
follows Proposition~\ref{prop:main}.

\begin{proposition}[General Cascade Identity]
Under Assumptions~1$'$ and~2$'$, for any allocation mechanism with $K$
fixed-supply goods,
\begin{equation}
\bfT = \bfbeta.
\end{equation}
The 2SLS coefficient $\beta_k$ equals the total societal effect of
expanding the supply of good~$k$ by one unit, including all equilibrium
reallocations through the mechanism.
\end{proposition}

The result holds because the 2SLS estimand inverts the full first-stage
matrix, and the first-stage matrix is exactly the object that governs
how the mechanism reallocates goods when supply changes. In a queue
system, $\Pi$ encodes how lottery outcomes redistribute students across
programs. In a market, $\Pi$ encodes how price changes redistribute
goods across consumers. In either case, inverting $\Pi$ traces out the
full chain of adjustments that a supply expansion triggers: the cascade.
The institution determines the content of $\Pi$; the algebra that
converts $\Pi$ and the reduced form into a policy effect is the same
regardless.

The identity fails when the supply of some good~$j$ inside the system
responds endogenously to the expansion of good~$k$. If $\bar{Q}_j$ is
not fixed, then the mechanism need not fully offset the cross-effect
$\pi_{jk}$ through reallocation: part of the adjustment is absorbed by
a change in total quantity rather than a redistribution among
individuals. The cascade equation then overstates the reallocation
component and understates the quantity-adjustment component, and
$\bfT \neq \bfbeta$. This is why the fixed-supply condition is the essential
assumption, not the details of the allocation mechanism.

\pagebreak

\subsection{The Neumann series and the infinite cascade: a two-program illustration}
\label{app:neumann}

To make the connection between the matrix inversion in Proposition~\ref{prop:main} and the dynamic cascade concrete, we spell out the Neumann series for the simplest non-trivial case: two programs with mutual substitution.

\paragraph{Setup.}
With $K = 2$, the cascade equation~\eqref{eq:cascade} becomes
\begin{equation}
T_1 = \frac{\mathrm{RF}_1}{\pi_{11}} + \frac{-\pi_{21}}{\pi_{11}} \, T_2, \qquad
T_2 = \frac{\mathrm{RF}_2}{\pi_{22}} + \frac{-\pi_{12}}{\pi_{22}} \, T_1.
\label{eq:cascade_2x2}
\end{equation}
Define $W_k = \mathrm{RF}_k / \pi_{kk}$, the instrument-specific Wald ratio for program~$k$, and
\begin{equation}
r_{21} = \frac{-\pi_{21}}{\pi_{11}}, \qquad r_{12} = \frac{-\pi_{12}}{\pi_{22}},
\end{equation}
the \emph{vacancy creation rates}: $r_{21}$ is the number of vacancies created in program~2 per new admission to program~1, and $r_{12}$ vice versa. Under substitution, $\pi_{21} < 0$ and $\pi_{12} < 0$, so both rates are positive. The system simplifies to
\begin{equation}
T_1 = W_1 + r_{21} \, T_2, \qquad T_2 = W_2 + r_{12} \, T_1.
\label{eq:cascade_simple}
\end{equation}

\paragraph{Iterated substitution.}
Rather than solving the system directly, substitute repeatedly to trace the cascade round by round. Starting from $T_1$:
\begin{align}
T_1 &= W_1 + r_{21}\bigl(W_2 + r_{12} \, T_1\bigr)
    = W_1 + r_{21} W_2 + r_{21} r_{12} \, T_1.
\end{align}
Instead of solving, substitute once more:
\begin{align}
T_1 &= W_1 + r_{21} W_2 + r_{21} r_{12}\bigl(W_1 + r_{21} \, T_2\bigr) \notag \\
    &= W_1 + r_{21} W_2 + r_{21} r_{12} \, W_1 + (r_{21} r_{12}) \, r_{21} \, T_2.
\end{align}
Continuing this process, a pattern emerges. Each pair of substitutions generates a factor of $\rho = r_{21} r_{12}$---one round trip through both programs---and the cascade alternates between refilling program~1 and program~2:
\begin{align}
T_1 &= \underbrace{W_1 + \rho \, W_1 + \rho^2 W_1 + \cdots}_{\text{refills in program 1}}
     \;+\; \underbrace{r_{21} W_2 + \rho \, r_{21} W_2 + \rho^2 r_{21} W_2 + \cdots}_{\text{refills in program 2}} \notag \\[6pt]
    &= W_1 \sum_{n=0}^{\infty} \rho^n + r_{21} W_2 \sum_{n=0}^{\infty} \rho^n
     = \frac{W_1 + r_{21} W_2}{1 - \rho}.
\label{eq:T1_series}
\end{align}
The cascade for program~2 follows by symmetry:
\begin{equation}
T_2 = \frac{W_2 + r_{12} W_1}{1 - \rho}.
\label{eq:T2_series}
\end{equation}
Both series converge when $\rho = r_{21} r_{12} < 1$: each round trip through the system displaces strictly fewer people than the previous one, so the cascade attenuates and terminates in the limit.

The cascade reads round by round as follows. Consider the expansion of program~1:
\begin{enumerate}[nosep]
\item[\emph{Round 0.}] Program~1 admits one new student. Direct effect: $W_1$. This creates $r_{21}$ vacancies in program~2.
\item[\emph{Round 1.}] Program~2 refills. Effect: $r_{21} W_2$. This creates $r_{21} r_{12} = \rho$ vacancies back in program~1.
\item[\emph{Round 2.}] Program~1 refills. Effect: $\rho \, W_1$. This creates $\rho \, r_{21}$ vacancies in program~2.
\item[\emph{Round 3.}] Program~2 refills. Effect: $\rho \, r_{21} W_2$. This creates $\rho^2$ vacancies in program~1.
\end{enumerate}
Each round trip attenuates by the factor $\rho$, and the total effect sums the contributions from all rounds.

\paragraph{Connection to the matrix formulation.}
Define the vacancy creation matrix
\begin{equation}
M = \begin{pmatrix} 0 & r_{21} \\ r_{12} & 0 \end{pmatrix},
\end{equation}
which collects the vacancy creation rates with zeros on the diagonal. The cascade equation~\eqref{eq:cascade_simple} can be rewritten as $(I - M)\,\mathbf{T} = \mathbf{W}$, so that
\begin{equation}
\mathbf{T} = (I - M)^{-1}\, \mathbf{W}.
\end{equation}
The Neumann series expands the inverse as
\begin{equation}
(I - M)^{-1} = \sum_{n=0}^{\infty} M^n = I + M + M^2 + M^3 + \cdots
\label{eq:neumann}
\end{equation}
provided the spectral radius of $M$ is less than one, which in the $2 \times 2$ case reduces to $\rho = r_{21} r_{12} < 1$. Computing the first few powers:
\begin{align}
M^0 &= \begin{pmatrix} 1 & 0 \\ 0 & 1 \end{pmatrix}, &
M^1 &= \begin{pmatrix} 0 & r_{21} \\ r_{12} & 0 \end{pmatrix}, \\[6pt]
M^2 &= \begin{pmatrix} r_{21} r_{12} & 0 \\ 0 & r_{12} r_{21} \end{pmatrix}, &
M^3 &= \begin{pmatrix} 0 & (r_{21} r_{12})\,r_{21} \\ (r_{12} r_{21})\,r_{12} & 0 \end{pmatrix}.
\end{align}
Even powers of $M$ are diagonal: the cascade has completed a round trip and returned to the program where it started. Odd powers are off-diagonal: the cascade is currently on the other side. All entries are non-negative (under substitution), so every term in the Neumann series contributes positively---there are no alternating signs. Summing the geometric matrix series:
\begin{equation}
\sum_{n=0}^{\infty} M^n = \frac{1}{1 - \rho} \begin{pmatrix} 1 & r_{21} \\ r_{12} & 1 \end{pmatrix}.
\end{equation}
Multiplying by $\mathbf{W} = (W_1, W_2)^T$ recovers exactly the scalar expressions~\eqref{eq:T1_series}--\eqref{eq:T2_series}.

\subsection{Monte Carlo verification in a dynamic admissions environment}
\label{app:simulation}
 
The cascade identity is an algebraic result that holds in population under Assumptions~\ref{ass:relevance} and~\ref{ass:margin}. This appendix verifies computationally that the identity survives in a dynamic admissions environment incorporating multiple rounds, merit-score retakes, programme switching, new cohorts, non-compliance, heterogeneous treatment effects, and monotonicity violations.\footnote{The Python code reproducing all simulation results is available upon request.}
 
\paragraph{Design.}
We simulate a system with 100 identical A-type and 100 identical B-type programmes, each admitting 15 students per round over 5 admission rounds. The initial applicant pool comprises 10{,}000 individuals; 3{,}000 new entrants join each subsequent round, approximately replacing the number admitted. Each individual is characterised by a coarse integer merit score (1--5, mimicking the granularity of the Swedish \emph{Högskoleprovet}), a stochastic programme preference, and heterogeneous potential outcomes: $Y_i(A) = Y_i(0) + 2 + 0.5\,a_i + \eta_i^A$ and $Y_i(B) = Y_i(0) + 3 + 0.3\,a_i + \eta_i^B$, where $a_i$ is latent ability and $\eta_i^A, \eta_i^B$ are idiosyncratic shocks.
 
Within each programme, applicants are ranked by merit with a lottery tiebreaker. Slots are offered sequentially down the ranked list. Non-defier applicants reject offers with probability 0.3 (random non-compliance), and rejected slots pass to the next applicant in the queue. Between rounds, remaining applicants may improve their merit score by one step (probability 0.3) or switch preferred programme type (probability 0.03).
 
Five percent of individuals are ``Groucho Marx defiers'': if admitted through the lottery (pivotal group, merit equal to cutoff), they reject the offer and \emph{permanently exit} the pool---they are never admitted. If they \emph{lose} the lottery, they deterministically reapply to the same programme with a guaranteed merit improvement, making eventual above-cutoff admission nearly certain. This produces a severe monotonicity violation: a lottery win permanently destroys the defier's admission prospects, while a lottery loss leads to almost certain admission. Among defier applicants, the individual-level correlation between the lottery and ever-admitted is $-0.033$, compared to $+0.086$ for non-defiers.
 
\paragraph{Estimation and test.}
The data are stacked as applicant $\times$ round-of-application, with each row contributing that round's lottery draw as the instrument, mirroring the structure of the actual Swedish admissions data. The 2SLS regression is $Y_i = \beta_0 + \beta_A\,\texttt{ever\_A}_i + \beta_B\,\texttt{ever\_B}_i + \delta_g + \varepsilon_i$, where $\delta_g$ is a pivotal group fixed effect (simplified to a programme-of-application indicator in the simulation).
 
The true societal effect $T_A$ is computed by expanding A-programme~0 by one slot in round~1, re-running the full multi-round simulation with identical pre-drawn random numbers, and computing $T_A = \sum_i Y_i^{\text{expanded}} - \sum_i Y_i^{\text{baseline}}$. A critical design choice ensures that all random draws (lottery, programme choice, retake, acceptance) are pre-drawn per individual per round before any simulation runs, indexed by person~ID. This prevents a contamination problem: if random numbers were drawn on-the-fly from the current pool, removing one person through the extra slot would shift subsequent draws for everyone else, generating thousands of spurious outcome changes from a single-slot expansion.
 
For each of 1{,}000 Monte Carlo replications, we draw fresh random numbers, estimate $\hat{\beta}_k$ by 2SLS, and compute $T_k$ from the paired expansion experiment using the same draws. The identity holds in population if $E[\hat{\beta}_k] = E[T_k]$, tested via the paired $t$-statistic $t_k = \bar{d}_k / (\text{SD}(d_k)/\sqrt{S})$ where $d_k^{(s)} = \hat{\beta}_k^{(s)} - T_k^{(s)}$.
 
\paragraph{Results.}
Table~\ref{tab:mc_results} reports the main results. Both paired $t$-statistics are well within the acceptance region, confirming $\bfT = \bfbeta$.
 
\begin{table}[h]
\centering
\caption{Monte Carlo verification of the cascade identity}
\label{tab:mc_results}
\begin{tabular}{lcccccc}
\toprule
& $\bar{\hat{\beta}}_k$ & $\bar{T}_k$ & $\bar{W}_k$ & Cascade & Diff & Paired $t$ \\
\midrule
A-type & 1.830 & 1.788 & 0.718 & 1.112 & 0.042 & 1.35 \\
B-type & 2.874 & 2.855 & 2.163 & 0.711 & 0.019 & 0.60 \\
\bottomrule
\end{tabular}
 
\medskip
{\small\textit{Notes:} 1{,}000 MC replications. 5 rounds, 30\% non-compliance, 5\% defiers, heterogeneous effects, retakes, switching, and new cohorts. $\bar{W}_k$ = mean Wald ratio. Cascade $= \bar{\hat{\beta}}_k - \bar{W}_k$. Diff $= \bar{\hat{\beta}}_k - \bar{T}_k$. Mean first-stage: $\bar{\pi}_{AA} = 0.142$, $\bar{\pi}_{BB} = 0.142$. ATE(A) $= 2.0$, ATE(B) $= 3.0$.}
\end{table}
 
The cascade is quantitatively important: for A-type programmes, 61\% of the total policy effect ($1.112/1.830$) is due to downstream reallocations rather than the direct effect on the marginal entrant.
 
Table~\ref{tab:mc_robust} reports paired $t$-statistics from stripped-down specifications that isolate each complication. The identity holds throughout: across 1 to 10 rounds, with and without non-compliance, and with defiers.
 
\begin{table}[h]
\centering
\caption{Paired $t$-statistics across specifications}
\label{tab:mc_robust}
\begin{tabular}{lcc}
\toprule
Specification & $t_A$ & $t_B$ \\
\midrule
Static (1 round), full compliance, no defiers & $-$0.84 & 0.58 \\
3 rounds, full compliance, no defiers & 0.25 & $-$0.61 \\
5 rounds, full compliance, no defiers & $-$0.44 & $-$0.15 \\
10 rounds, full compliance, no defiers & 1.27 & $-$0.07 \\
5 rounds, 30\% non-compliance, no defiers & $-$0.30 & 0.94 \\
5 rounds, 30\% non-compliance, 5\% defiers & 1.35 & 0.60 \\
\bottomrule
\end{tabular}
 
\medskip
{\small\textit{Notes:} Stripped-down specifications use 200 MC replications; full specification (last row) uses 1{,}000. All specifications include heterogeneous treatment effects. Multi-round specifications include retakes, switching, and new cohorts.}
\end{table}
 
The identity fails in one predictable case: when programmes are \emph{under-subscribed} (slots exceed willing acceptors after accounting for non-compliance). The extra slot then has nothing to fill, yielding $T_k = 0$, while $\hat{\beta}_k$ is unaffected because the first stage is zero at under-subscribed programmes. This confirms that the binding condition is Assumption~\ref{ass:margin}: programmes must be capacity-constrained.

\subsection{The cascade identity in potential outcomes---a three-program example}
\label{sec:threeprog}

To build further intuition, we work through the simplest non-trivial case in a potential outcomes framework: three mutually exclusive choices indexed $j \in \{0, 1, 2\}$, where $j=0$ is the outside option (no higher education), $j=1$ is a mid-tier program, and $j=2$ is a more competitive program. Each program has its own admission lottery. We ask: what does the 2SLS coefficient on program~$j$ actually measure, and does it equal the societal policy effect $T_j$?

Following the potential outcomes framework used in e.g.\ \citet{kirkeboen2016}, a student's outcome is
\begin{equation}
Y = Y_0 + (Y_1 - Y_0)\,\mathbf{1}(j=1) + (Y_2 - Y_0)\,\mathbf{1}(j=2),
\end{equation}
where $Y_0, Y_1, Y_2$ are the potential outcomes under each choice. The 2SLS second stage is
\begin{equation}
Y = \beta_0 + \beta_1 A_1 + \beta_2 A_2 + \varepsilon,
\end{equation}
with instruments $Z_1$ and $Z_2$ for each program's lottery.

To make the cascade tractable by hand in this three-program example, we
impose
\begin{equation}
\pi_{21} = 0.
\label{eq:ranked_simple}
\end{equation}
This restriction says that the program~1 lottery does not affect
program~2 enrollment. It is akin to the irrelevance condition in \citet{kirkeboen2016}. It is plausible in a one-shot application game
where programs can be ordered by selectivity: a student on the margin
for the less competitive program~1 is typically far from the admission
threshold for program~2. In our data, which span 2008--2021, the
restriction need not hold---a student who loses the program~1 lottery
could retake the Swedish Scholastic Aptitude Test and eventually gain
admission to program~2. We adopt it here purely because it collapses
the cascade to a single step and allows the full argument to be traced easily
by hand. The general result (Proposition~\ref{prop:main}) does not
require it: the cascade identity $\beta_k = T_k$ holds for any
first-stage matrix, including cases where $\pi_{21} \neq 0$. Under the restriction, the only transitions induced by the two lotteries are:

\begin{center}
\begin{tabular}{ll}
$0 \to 1$ & a program~1 lottery win admits a student from outside higher education \\
$0 \to 2$ & a program~2 lottery win admits a student from outside higher education \\
$1 \to 2$ & a program~2 lottery win admits a student who would otherwise have attended program~1
\end{tabular}
\end{center}

\paragraph{Solving the moment conditions.}
The two reduced-form equations are
\begin{align}
\mathrm{RF}_1 &= \beta_1 \pi_{11} + \beta_2 \pi_{21} = \beta_1 \pi_{11}, \label{eq:rf1} \\
\mathrm{RF}_2 &= \beta_1 \pi_{12} + \beta_2 \pi_{22}, \label{eq:rf2}
\end{align}
where the simplification in~\eqref{eq:rf1} uses the restriction $\pi_{21} = 0$. Since the program~1 lottery moves students only between the outside option and program~1, the first equation yields a clean Wald ratio:
\begin{equation}
\beta_1 = \frac{\mathrm{RF}_1}{\pi_{11}} = E[Y_1 - Y_0 \mid 0 \to 1].
\end{equation}
This is a standard LATE: the average gain from attending program~1 relative to no higher education, for students at the program~1 admission margin.

For program~2, solving~\eqref{eq:rf2} gives
\begin{equation}
\beta_2 = \frac{\mathrm{RF}_2}{\pi_{22}} - \frac{\pi_{12}}{\pi_{22}} \, \beta_1.
\label{eq:beta2_decomp}
\end{equation}
The first term, $\mathrm{RF}_2 / \pi_{22}$, is the instrument-specific Wald ratio for program~2: the direct effect per new admit, ignoring cross-effects. The second term is the cascade correction. Because program~2 lottery wins draw some students away from program~1, we have $\pi_{12} < 0$, so $-\pi_{12}/\pi_{22} > 0$ is the rate at which new program~2 admits vacate seats in program~1. Each such vacancy is filled from program~1's ranking list, generating a societal gain of $\beta_1$. The cascade correction thus adds the downstream effect of refilling the seats left behind.

\paragraph{Expressing the coefficients in terms of complier types.}
We can now substitute the complier-type expressions for each component. The program~2 lottery induces two transitions: $0 \to 2$ with probability $p_{02}$ and $1 \to 2$ with probability $p_{12}$. These give:
\begin{align}
\mathrm{RF}_2 &= p_{02}\,E[Y_2 - Y_0 \mid 0 \to 2] + p_{12}\,E[Y_2 - Y_1 \mid 1 \to 2], \\
\pi_{22} &= p_{02} + p_{12}, \\
\pi_{12} &= -p_{12}.
\end{align}
The last expression reflects that each $1 \to 2$ transition reduces program~1 enrollment by one. Substituting into~\eqref{eq:beta2_decomp}:
\begin{equation}
\beta_2 = \frac{p_{02}\,E[Y_2 - Y_0 \mid 0 \to 2] \;+\; p_{12}\,\bigl(E[Y_2 - Y_1 \mid 1 \to 2] + E[Y_1 - Y_0 \mid 0 \to 1]\bigr)}{p_{02} + p_{12}}.
\label{eq:beta2_simple}
\end{equation}
The numerator makes the cascade visible at the level of individual treatment effects. When the new program~2 seat is filled directly from outside higher education (probability $p_{02}$), society gains $Y_2 - Y_0$ for that student. When it is filled by a student previously enrolled in program~1 (probability $p_{12}$), there are two gains: the student who moves up gains $Y_2 - Y_1$, and the person who fills the vacated program~1 seat from outside higher education gains $Y_1 - Y_0$.

Expression~\eqref{eq:beta2_simple} is the cascade identity. Expanding program~2 by one seat sets off the following chain:
\begin{enumerate}[nosep]
\item With probability $p_{02}/(p_{02}+p_{12})$, the new seat is filled directly from outside higher education. Society gains $Y_2 - Y_0$ for that student.
\item With probability $p_{12}/(p_{02}+p_{12})$, the new seat is filled by a student previously enrolled in program~1. Society gains $Y_2 - Y_1$ for that student. But the vacated program~1 seat is then filled from outside higher education, generating a further gain of $Y_1 - Y_0$.
\end{enumerate}
The cascade terminates after one step because restriction~\eqref{eq:ranked_simple} rules out any further displacement: the program~1 lottery does not affect program~2 enrollment, so the refilled program~1 seat cannot trigger further substitution. The 2SLS coefficient $\beta_2$ is exactly the per-seat societal effect of this expansion, confirming Proposition~\ref{prop:main} in this setting.

Without restriction~\eqref{eq:ranked_simple}, a program~1 lottery win could draw students away from program~2, generating a further round of substitution when program~2 refills its vacancy. The cascade would then require summing an infinite series---which is what the matrix inversion in Proposition~\ref{prop:main} accomplishes in general. The ranked restriction collapses this to a one-step cascade and allows the full argument to be traced by hand.

\end{document}